\documentclass[10pt,a4paper,twoside]{article}

	\usepackage[hyphens]{url}				
	\usepackage{hyperref}									
		\hypersetup{colorlinks,linkcolor={green!30!black},citecolor={blue!50!black},urlcolor={blue!80!black}}
		\hypersetup{breaklinks=true}		
	\usepackage{mathtools,amssymb,amsthm}						
	\usepackage{dsfont}										
	\usepackage[
		backend=biber,style=authoryear-icomp,natbib=true,
		sortcites=true,ibidtracker=strict,sorting=nyvt,giveninits=true,
		uniquename=init,maxbibnames=99]{biblatex}				
		\renewbibmacro{in:}{}								
			\newbibmacro{string+doi}[1]{%
  			\iffieldundef{doi}{#1}{\href{http://dx.doi.org/\thefield{doi}}{#1}}}
			\DeclareFieldFormat{title}{\usebibmacro{string+doi}{\mkbibemph{#1}}}
			\DeclareFieldFormat[article]{title}{\usebibmacro{string+doi}{\mkbibquote{#1}}}
	
	\usepackage[hyphenbreaks]{breakurl}		
	\usepackage[
		top=2.8cm,bottom=2.8cm,
		inner=2.5cm,outer=2.5cm,
		a4paper]{geometry}									
	\usepackage{tikz}
	\usepackage[font=small,labelfont=bf,centerlast]{caption}			
	\usepackage{array}										
		\newcolumntype{L}{>{$}l<{$}} 							
		\newcolumntype{C}{>{$}c<{$}} 						
		\newcolumntype{R}{>{$}r<{$}} 							
	
	\makeatletter
		\AtBeginDocument{\toggletrue{blx@useprefix}}
		\AtBeginBibliography{\togglefalse{blx@useprefix}}
	\makeatother

	\newcommand{\ncom}{\newcommand}
	\ncom{\h}{\mathcal{H}}
	\ncom{\een}{\mathds{1}}
	\ncom{\ket}[1]{\left|#1\right\rangle}
	\ncom{\bra}[1]{\left\langle#1\right|}
	\ncom{\braket}[2]{\left\langle#1\middle|#2\right\rangle}
	\ncom{\ketbra}[2]{\left|#1\middle\rangle\middle\langle#2\right|}
	\ncom{\expv}[3]{\left\langle#1\middle|#2\middle|#3\right\rangle}
	\ncom{\bket}[6]{\left|#1,#2\middle|#3,#4\middle|#5,#6\right\rangle}
	\ncom{\proj}[6]{\left\lgroup#1,#2\middle|#3,#4\middle|#5,#6\right\rgroup}
	\ncom{\set}[2]{\left\{#1\:\middle|\:#2\right\}}				
	\ncom{\dee}{\mathrm{d}}								
	\ncom{\ldef}{\coloneqq}								
	\ncom{\rdef}{\eqqcolon}								
	\ncom{\Epsilon}{\mathcal{E}}							
	\ncom{\pee}{\mathds{P}}								
	\ncom{\peetje}{\mathrm{p}}							
	\ncom{\ee}{\mathds{E}}								
	\ncom{\var}{\mathds{V}\mathrm{ar}}					
	\ncom{\cov}{\mathds{C}\mathrm{ov}}					
	\DeclarePairedDelimiter{\ceil}{\lceil}{\rceil}

	\theoremstyle{plain}	
		\newtheorem{theorem}{Theorem}
		\newtheorem*{theorem*}{Theorem}
		\newtheorem{proposition}{Proposition}
		
		\newtheorem{claim}{Claim}\ncom{\claimautorefname}{Claim}	
	\theoremstyle{definition}
		\newtheorem{definition}{Definition}\ncom{\definitionautorefname}{Definition}

\begin{document}

\title{Completely real?\\ A critical note on the claims by Colbeck and Renner} 
\author{R. Hermens}
\date{November 04, 2020}

\maketitle

\begin{abstract}
In a series of papers \citet{ColbeckRenner11,ColbeckRenner15Comp,ColbeckRenner15Suff} claim to have shown that the quantum state provides a complete description for the prediction of future measurement outcomes.
In this paper I argue that thus far no solid satisfactory proof has been presented to support this claim.
Building on the earlier work of \citet{Leifer14}, \citet{Landsman15} and \citet{Leegwater16}, I present and prove two results that only partially support this claim.
I then discuss the arguments by Colbeck, Renner and Leegwater concerning how these results are to generalize to the full claim.
This argument turns out to hinge on the implicit use of an assumption concerning the way unitary evolution is to be represented in any possible completion of quantum mechanics. 
I argue that this assumption is unsatisfactory and that possible attempts to validate it based on measurement theory also do not succeed.
\end{abstract}

\section{Introduction}
Can quantum-mechanical description of physical reality be considered complete?
In the famous paper with this title, \citet{EPR35} argued that the question should be answered in the negative.
It was one of several arguments that Einstein devised and although it was presumably among his least favorites \citep{Fine17}, it is still the most widely known.

Around the same time, \citet{Neumann27,Neumann32} presented a formal argument towards the opposite conclusion.
Apart from starting from different assumptions\footnote{See \citep{Bub10,Dieks17} for comprehensive accounts of von Neumann's proof.}, the adopted notions of completeness are also quite distinct \citep{ElbyBrownFoster93}.
While Einstein was concerned with whether the quantum mechanical description sufficed to give a physical \emph{explanation} of the phenomena predicted, von Neumann adopted a more operational approach concerning the question whether the addition of hidden variables could allow for deviating \emph{predictions} for the phenomena.
More precisely, the following question was considered.
If we consider an ensemble of systems $E$ described by a pure quantum state $\psi$, is it possible to decompose this ensemble into sub-ensembles $E_1,E_2,\ldots$ such that the predictions for the sub-ensembles are not equal to the predictions for the total ensemble?
If not, then quantum mechanics may be considered complete.\footnote{Actually, von Neumann was specifically considering the possibility of dispersion-free sub-ensembles.}

Colbeck and Renner's completeness claim \citep{ColbeckRenner11,ColbeckRenner15Comp,ColbeckRenner15Suff} alludes to von Neumann's notion of completeness.
In their own words, they show that ``[u]nder the assumption that measurements can be chosen freely [...] no extension of quantum theory can give more information about the outcomes of future measurements than quantum theory itself'' \citep[p. 1]{ColbeckRenner11}.
In a follow-up paper \citet{ColbeckRenner12} extended their proof to obtain a $\psi$-ontology theorem.\footnote{The most famous $\psi$-ontology theorem is the PBR theorem \citep{PuseyEtAl12}. See \citep{Leifer14} for an overview of $\psi$-ontology theorems.}
Making use of the same ``free choice'' assumption, they conclude that ``a system's wave function is in one-to-one correspondence with its elements of reality.''

The assumption of ``free choice'' has since been identified as the conjunction of two more familiar assumptions: parameter independence and setting independence\footnote{This second assumption states that the settings for a measurement are independent of the state of the system to be measured. It also appears in the literature as measurement independence, $\lambda$-independence, no conspiracy, free will, and presumably under several other names.} \citep{GhirardiRomano13,VonaLang14}.
Despite this clarification, there has been confusion about whether these two assumptions suffice, or if more assumptions that rely on the specific mathematical structure of quantum mechanics are needed.
Since then, \citet{Leifer14}, relying on \citep{ColbeckRenner17}, presented a rigorous proof for the second claim by Colbeck and Renner that does not rely on the validity of the completeness claim.
\citet{Landsman15} gave a critical assessment of the completeness claim, arguing that on top of these explicit assumptions, the proof for $\psi$-completeness relies on no less than four further rather technical assumptions.
A far more friendly conclusion was reached by \citet{Leegwater16} who gave a thorough reworking of Colbeck and Renner's original proof.
However, Landsman's worries were not explicitly addressed by Leegwater, and the proof is not transparent enough to easily assess whether Landsman's criticism was indeed moot.

In this paper I argue that the general conclusion drawn by Colbeck and Renner is currently unwarranted.
To this end a formal statement of their two claims in terms of the ontic models framework is given in \autoref{FormalClaimsSec} after motivating the use of this framework in \autoref{frameworksec}.
In \autoref{FormalClaimsSec} it is also shown that the completeness claim is logically stronger than the $\psi$-ontology claim, thus showing that the issues of \citet{Landsman15} are not simply resolved by reworking the proof of \citet{Leifer14}.
The general strategy for proving the completeness claim is to start with proving it for the special case of two systems in a maximally entangled state.
This case is discussed in \autoref{partialsec} as well as its generalization to arbitrary entangled states.
In \autoref{qubitsec} I discuss how this result is supposed to generalize to the case of individual systems.
The crucial step needed to make that generalization is then scrutinized and criticized in \autoref{critsec} and I conclude that no satisfactory proof of the completeness claim is currently available.

\section{The use of ontic models}\label{frameworksec}
Ontic models are a useful tool for studying non-classical features of quantum mechanics.\footnote{The framework was first introduced in \citep{Spekkens05}. See also \citep{Harrigan07}.}
An ontic model for a system introduces a set $\Lambda$ of possible ontic states for the system.
These ontic states determine how the system is to respond in the case of a measurement, in the sense that with each possible measurement procedure the state $\lambda\in\Lambda$ associates a probability distribution over possible measurement outcomes.
These probability distributions may be thought of as representing, possibly dispositional, properties of the system.
Possible preparations of the system are associated with probability distributions over the state space (thus $\Lambda$ is assumed to be a measurable space).
These may be thought of as representing ignorance concerning the ``true'' state of the system $\lambda$.

\subsection{Motivating their use}

The use of ontic models is not trivially innocuous.
A possible reason for disliking ontic models, is that they seem to come equipped with metaphysical baggage.
It is common, as suggested by their name, to interpret ontic states as representing the physical properties of the system and with a slight abuse of language these states may be called ``real''.\footnote{See \citep{Halvorson19} for further discussion.}
Colbeck and Renner indeed seem hesitant to embrace this kind of terminology as indicated for example by their short discussion of the PBR theorem \citep[p. 68]{ColbeckRenner15Suff}.
But no such interpretation is forced upon us.
If one wants, one can simply think of ontic states as mathematical objects competing with quantum states for being ``the best'' description of the system.
The interpretation of ontic states for an ontic model is then as much open to debate as the interpretation of quantum states.

Another possible reason for not wanting to use ontic models is that they may not be general enough.
Indeed, setting independence is an implicit assumption of this framework since probability distributions over ontic states are taken to not depend on which measurements are or are not performed on the system.
This is in fact a common loophole in no-go theorems that may be exploited in, for example, retrocausal approaches \citep{EvansFriederich19}.
Although there are ways to generalize the framework to try to accommodate for this loophole\footnote{See for example \citep{Hermens19}.}, there is no need to go into this issue since setting independence is accepted as part of the assumptions in Colbeck and Renner's claims.

With respect to the issue of setting dependence, the framework used by Colbeck and Renner is indeed more general.
The framework they use is more akin to the use of causal networks such as in the work of \citet{WoodSpekkens15}.\footnote{It deserves to be noted though that the approach of Colbeck and Renner is rather unorthodox. The nodes in their networks are, what they call, ``spacetime random variables'' and their causal arrows have a meaning that is quite distinct from the one adopted in the theory of causal networks.}
In this approach, measurement settings, measurement outcomes, quantum states and hidden variable states are all treated on a par as random variables on some big probability space. 
So correlations between states and settings are indeed allowed.

However, the network approach is not general enough for the purposes of Colbeck and Renner's proof.
What is needed is that probabilities for outcomes conditional on settings and states are always well-defined. 
Within the network approach this meams that probabilities for settings and states are always well-defined and non-zero, so they can be conditionalized upon.
When considering a finite fragment of quantum mechanics (finite number of states and settings), such as in Bell's theorem, such probabilities can always be defined.
But Colbeck and Renner don't stick to the finite setting in their proof and it is doubtful whether the claims of Colbeck and Renner (possibly modified) could be proven in such a finite setting.\footnote{This can be seen from the fact that in the EPRB setting with two possible settings for both Alice and Bob any maximally entangled state can be written as a convex combinations of PR boxes. Thus within that finite setting the completeness claim does not hold.}

Within an ontic model, the probability of an outcome conditional on a setting is well-defined and is taken to be primitive.
A price to pay is that settings and outcomes are not treated on a par and probabilities for settings are not defined.
One could insist that, fundamentally, there should be no distinction between the two. 
But this is not a complaint against ontic models specifically, but should be considered on a more general level.
As \citet[p. 438]{SeevinckUffink11} note:
\begin{quote}
Even quantum mechanics leaves the question what measurement is going to be performed on a system as one that is decided \emph{outside} the theory, and does not specify how much more probable one measurement is than another.
It thus seems reasonable not to require from the candidate theories that they describe such probabilities.
\end{quote}
In this light we see that the complaint against ontic models is actually not that they are not general enough because they allow a distinction between settings and outcomes but, actually, they are too lenient for having a distinction that allows for ambiguity concerning the true ontology.
But that could hardly be a complaint for using the framework to prove theorems in the foundations of quantum mechanics.

For the reader not persuaded by any of these considerations, there is one final motivation for the use of ontic models.
I currently have no other viable framework within which I can formulate a rigorous partial proof for the claims by Colbeck and Renner.
It is simply a requirement that the conditional probabilities are well-defined, and ontic models simply seem to provide a minimal mathematical structure needed to establish this.
In this regard, it also deserves to be mentioned that ``ontic model'' is not much more than a label for a standardized mathematical structure that has been (sometimes implicitly) used to analyze results like Bell's theorem for decades.


\subsection{Formalism}

The prime constraint for ontic models, is that they can reproduce the quantum mechanical predictions for measurements on any quantum system.
Here, with a quantum system we associate a finite-dimensional Hilbert space $\h$ together with a triple of sets $(\mathcal{P},\mathcal{M},\mathcal{T})$ denoting sets of preparations, measurements and transformations respectively.
It is assumed that every preparation $P\in\mathcal{P}$ can be represented by a density operator $\rho$, every measurement $M\in\mathcal{M}$ can be represented by a self-adjoint operator $A$ (i.e., only PVMs are considered), and every transformation $T\in\mathcal{T}$ can be represented by a unitary operator $U$.
If $\rho$ is pure, it will often be represented by a unit vector $\psi$ that satisfies $\rho=[\psi]$, where $[\psi]$ is the 1-dimensional projection on the line spanned by $\psi$.
Taking into account the work of \citet{Spekkens05}, the triple $(\mathcal{P},\mathcal{M},\mathcal{T})$ is expected to be contextual, i.e, the mappings $M\mapsto A$, $P\mapsto\rho$, $T\mapsto U$ will in general be many-to-one.
For the sake of simplicity, it is assumed that the mappings are also onto.
Thus every quantum state can be prepared, every self-adjoint operator can be measured, and every unitary operation can be brought about by means of a transformation.\footnote{If one is interested in the possibility of experimental tests, a focus on theorems that can be formulated for finite triples is preferable (as for example in Bell's theorem), since the validity of quantum mechanics can only be investigated for at most a finite number of predictions. In this paper the issue is rather the validity of a particular claim. So it makes sense to work with one of the weakest formulations of the claim relying on the validity of quite a lot of quantum mechanics.}
The probabilities for measurement outcomes are given by the Born rule.
That is, when $P,T,M$ are represented by $\rho,U,A$ respectively, then
\begin{equation}
	\pee(a|M,T,P)=\mathrm{Tr}\left(U\rho U^*[a]_A\right),
\end{equation}
where $[a]_A$ is the projection onto the eigenspace of $A$ corresponding to the eigenvalue $a$.

An ontic model for a quantum system consists of a measurable space of ontic states $(\Lambda,\Sigma)$, where $\Lambda$ is the set of ontic states and $\Sigma$ is a $\sigma$-algebra of subsets of $\Lambda$.
Every measurement $M\in\mathcal{M}$, is associated with a Markov kernel $\peetje_M$, called a response function, that associates with every $\lambda\in\Lambda$ a probability distribution $\peetje_M(\:.\:|\lambda)$ over the possible measurement outcomes.
Following \citet{Leegwater16}, these probabilities will be called $\lambda$-probabilities.
Every preparation $P\in\mathcal{P}$ is associated with a probability measure $\mu_P$ over the ontic states and every transformation $T\in\mathcal{T}$ is associated with a Markov kernel $\gamma_T$ from $\Lambda$ to itself.
On average, the predictions of quantum mechanics are required to be reproduced:
\begin{equation}
\begin{split}
	\pee(a|M,P)&=\int\peetje_M(a|\lambda)\dee\mu_P(\lambda),\\
	\pee(a|M,T,P)&=\iint\peetje_M(a|\lambda)\gamma_T(\dee\lambda|\tilde{\lambda})\dee\mu_P(\tilde{\lambda}).
\end{split}
\end{equation}
Often, when there is no cause for confusion, the use of $P,T,M$ will be replaced by their quantum mechanical representatives, resulting in more transparent equations like
\begin{equation}
	\iint\peetje_A(a|\lambda)\gamma_U(\dee\lambda|\tilde{\lambda})\dee\mu_\psi(\tilde{\lambda})=\expv{U\psi}{[a]_A}{U\psi}.
\end{equation}
In other cases, the quantum representatives will be added as subscripts.
So $M_A$ denotes a measurement procedure represented by the self-adjoint operator $A$ in quantum mechanics.

It is worth noting that transformations can be reconsidered to be part of either the preparation procedure or the measurement procedure.
Specifically, for any $M_A,T_U$ there is a response function $\peetje_{M_A\circ T_U}$ defined by
\begin{equation}
	\peetje_{M_A\circ T_U}(a|\lambda)\ldef\int\peetje_{M_A}(a|\tilde{\lambda})\gamma_{T_U}(\dee\tilde{\lambda}|\lambda),
\end{equation}
which corresponds to an operational procedure for a measurement that is represented quantum mechanically by the operator $U^*AU$.
Likewise, for any $T_U,P_\rho$ there is a probability distribution $\mu_{T_U\circ P_\rho}$ defined by
\begin{equation}
	\mu_{T_U\circ P_\rho}(\Delta)\ldef\int\gamma_{T_U}(\Delta|\lambda)\dee\mu_{P_\rho}(\lambda),~\Delta\in\Sigma,
\end{equation}
which corresponds to an operational procedure for a preparation of the state represented by $U\rho U^*$.
Finally, any two transformations $\gamma_{T_1},\gamma_{T_2}$ may be stringed together to give the transformation ``$T_2$ after $T_1$'' given by
\begin{equation}
	\gamma_{T_2\circ T_1}(\Delta|\lambda)\ldef
	\int\gamma_{T_2}(\Delta|\tilde{\lambda})\gamma_{T_1}(\dee\tilde{\lambda}|\lambda).
\end{equation}

\section{\texorpdfstring{$\boldsymbol{\psi}$}{psi}-completeness and \texorpdfstring{$\boldsymbol{\psi}$}{psi}-ontology}\label{FormalClaimsSec}

A straightforward ontic model for quantum mechanics is the one by \citet{BeltramettiBugajski95}, which is basically quantum mechanics itself.
The set of ontic states $\Lambda$ is taken to be the set of one-dimensional projections.
The $\lambda$-probabilities for a measurement $M_A$ are given by the Born rule, i.e., for each $[\psi]\in\Lambda$
\begin{equation}
	\peetje_{M_A}(a|[\psi])=\expv{\psi}{[a]_A}{\psi}.
\end{equation}
A preparation $P_\psi$ of a pure state corresponds to the Dirac-distribution centered on $[\psi]$, while a preparation of a mixed state corresponds to an appropriate convex combination of such Dirac-distributions.
So in general there are multiple distinct distributions $\mu_\rho$ corresponding to the same $\rho$.
The quantum dynamics are just copied, i.e, 
\begin{equation}
	\gamma_U(\Delta|[\psi])=
	\begin{cases}
		1 & [U\psi]\in\Delta,\\
		0 & \text{otherwise.}
	\end{cases}
\end{equation}

This model may rightfully be said to be trivial.
In particular because it has the property that the $\lambda$-probabilities coincide with the quantum probabilities.
This is the key idea to the formal notion of triviality and $\psi$-completeness.
\begin{definition}\label{PsiCompDef}
An ontic model for a quantum system is said to be \emph{trivial} with respect to a set of measurements $\mathcal{M}'\subset\mathcal{M}$ and preparations $\mathcal{P}'\subset\mathcal{P}$ if for every preparation $P_\rho\in\mathcal{P}'$ and every measurement $M_A\in\mathcal{M}'$ the $\lambda$-probabilities $\mu_{P_\rho}$-almost surely coincide with the quantum mechanical probabilities:
\begin{equation}
	\int\left|\peetje_{M_A}(a|\lambda)-\mathrm{Tr}\left(\rho [a]_A\right)\right|\dee\mu_{P_\rho}(\lambda)=0
\end{equation}
or, equivalently,
\begin{equation}
	\mu_{P_{\rho}}\left(\set{\lambda\in\Lambda}{\peetje_{M_A}(a|\lambda)=\mathrm{Tr}\left(\rho [a]_A\right)}\right)=1.
\end{equation}
If the ontic model is trivial with respect to all measurements and preparations, then it is called \emph{$\psi$-complete}.
\end{definition}

Colbeck and Renner's completeness claim may now be formulated as follows:
\begin{claim}\label{claim1}
For any quantum system every ontic model that satisfies parameter independence is $\psi$-complete.
\end{claim}

It is worth pointing out that the notion of $\psi$-completeness was already introduced by \citet{HarriganSpekkens10} with a different, but formally related, meaning.
They call an ontic model $\psi$-complete if the set of ontic states is isomorphic to the set of ontic states of the Beltrametti-Bugajski model and preparations of pure quantum states correspond to the appropriate Dirac distributions.
It follows that if an ontic model is $\psi$-complete in the sense of Harrigan and Spekkens, it is automatically $\psi$-complete in the sense of \autoref{PsiCompDef}.

In the same paper, Harrigan and Spekkens introduced the distinction between $\psi$-ontic and $\psi$-epistemic ontic models.
Roughly, the idea is that whenever two preparations corresponding to two distinct pure quantum states, then in a $\psi$-ontic model these preparations will yield distinct ontic states for the system.
The following definition is the now standard way to make this idea precise.
\begin{definition}\label{PsiOnticDef}
An ontic model for a quantum system is said to be \emph{$\psi$-ontic} if for any pair of preparations $P_1,P_2$ represented by distinct pure quantum states the variational distance between their corresponding probability measures $\mu_{P_1},\mu_{P_2}$ is one, i.e.,
\begin{equation}\label{PsiOnticEq}
	\sup_{\Delta\in\Sigma}\left|\mu_{P_1}(\Delta)-\mu_{P_2}(\Delta)\right|=1.
\end{equation}
\end{definition}

Colbeck and Renner's $\psi$-ontology claim may now be formulated as follows:
\begin{claim}\label{claim2}
For any quantum system every ontic model that satisfies parameter independence is $\psi$-ontic.
\end{claim}
This claim and \autoref{claim1} above are intended to be paraphrases of Claim 1 and Claim 2 in \citep{ColbeckRenner15Comp} adapted to the framework of ontic models.
But it deserves to be noted that instead of looking at models for specific systems, Colbeck and Renner are looking at an ``alternative theory that is compatible with quantum theory''.
In the next section it will become clear that this is not a moot point since their argument relies on relating several quantum systems, which will then yield a slight modification of the two claims.

In the papers \citep{ColbeckRenner12,ColbeckRenner15Comp}, Colbeck and Renner proved Claim 2 as part of a corollary of a proof for Claim 1.
An independent proof for Claim 2 was given in \citep{ColbeckRenner17} which was also referred to in \citep{ColbeckRenner15Suff} and served as the main inspiration for Leifer's formulation and proof of the Colbeck-Renner $\psi$-ontology theorem in \citep{Leifer14}.
It may be noted that actually Leifer showed not that models should be $\psi$-ontic, but that the overlap between two quantum states tends to zero as the dimension of the system goes to infinity.
Thus Claim 2 only holds in the limit and we will see that something similar is needed when considering Claim 1.

The following proposition establishes the relation between the two claims more directly.
\begin{proposition}
Any ontic model for a quantum system that is $\psi$-complete is also $\psi$-ontic.
\end{proposition}
\begin{proof}
Suppose a quantum system is given together with a $\psi$-complete ontic model. 
Let $P_1,P_2$ be any two preparations corresponding to distinct pure quantum states $\psi_1,\psi_2$.
Then there exists a measurement $M\in\mathcal{M}$ corresponding to a self-adjoint operator $A$ with eigenvalue $a$ such that
\begin{equation}
	\expv{\psi_1}{[a]_A}{\psi_1}\neq\expv{\psi_2}{[a]_A}{\psi_2}.
\end{equation}
Because the model is $\psi$-complete it follows that
\begin{equation}
	\mu_{P_1}\left(\set{\lambda\in\Lambda_\h}{\peetje_M(a|\lambda)=\expv{\psi_1}{[a]_A}{\psi_1}}\right)=1,
\end{equation}
while
\begin{equation}
	\mu_{P_2}\left(\set{\lambda\in\Lambda_\h}{\peetje_M(a|\lambda)=\expv{\psi_1}{[a]_A}{\psi_1}}\right)=0
\end{equation}
so \eqref{PsiOnticEq} holds.
\end{proof}

The converse statement of this proposition is not true.
One can think of Bohmian mechanics as an example of a $\psi$-ontic theory that is not $\psi$-complete.
Thus there is little hope for finding an easy proof for Claim 1 based on Claim 2.
My focus is therefore on Claim 1 alone, which \citet{Landsman15} has criticized but which was endorsed by \citet{Leegwater16} after he clarified some of the technicalities behind the proof.
I shall argue that a rigorous proof for Claim 1 is currently lacking.
But before that, I shall discuss and prove some weaker versions of Claim 1 as a partial endorsement of Leegwater's conclusion.

\section{Partial proofs for the completeness claim}\label{partialsec}
\subsection{The equiprobability theorem}
The simplest partial proof for \autoref{claim1} concerns local measurements on a pair of $d$-level quantum systems in a maximally entangled state.
The question is then of course how to generalize this to arbitrary states, arbitrary measurements and arbitrary systems.
But first an unambiguous formulation and proof of the partial claim is required.

The scenario is the familiar EPR setup with associated Hilbert space $\h=\mathbb{C}^d\otimes\mathbb{C}^d$, and a preparation of the system in the state
\begin{equation}
	\psi_d\ldef\sum_{i=1}^d\frac{1}{\sqrt{d}}e_i\otimes e_i,
\end{equation}
where $(e_i)_{i=1}^d$ is an arbitrary orthonormal basis for $\mathbb{C}^d$.
One of the subsystems is then sent to Alice, and the other to Bob, who are assumed to be space-like separated.

Now let $\mathcal{M}_A\subset\mathcal{M}$ denote the set of possible measurements where only Alice performs a measurement (locally).
These can be represented by self-adjoint operators of the form $A\otimes\een$.
Symmetrically, let $\mathcal{M}_B$ be the set of possible measurements where only Bob performs a measurement, which can thus be represented by operators of the form $\een\otimes B$.
Finally, take $\mathcal{M}_\mathrm{LOC}$ to be the set of measurements where either Alice or Bob performs a measurement or both.
So one may take $\mathcal{M}_\mathrm{LOC}\simeq\mathcal{M}_A\times\mathcal{M}_B$.
Parameter independence can now be formulated as
\begin{equation}
\begin{split}
	\peetje_{M_{A\otimes\een}}(a|\lambda)&=\sum_b\peetje_{M_{A\otimes B}}(a,b|\lambda),\\
	\peetje_{M_{\een\otimes B}}(b|\lambda)&=\sum_a\peetje_{M_{A\otimes B}}(a,b|\lambda),
\end{split}
\end{equation}
for all $M_{A\otimes\een}\in\mathcal{M}_A,M_{\een\otimes B}\in\mathcal{M}_B,M_{A\otimes B}\in\mathcal{M}_{\mathrm{LOC}}$.
Here the local operational procedures $M_{A\otimes\een}$ and $M_{\een\otimes B}$ should be the same as those represented by $M_{A\otimes B}$.
One then has the following theorem.

\begin{theorem}[Equiprobability theorem]\label{equipthm}
Any ontic model for the pair of $d$-level systems that satisfies parameter independence must be trivial w.r.t. $\mathcal{M}_{\mathrm{LOC}}$ and all preparations that are represented by $\psi_d$.
\end{theorem}

The basic structure of the proof for this theorem is due to \citet{ColbeckRenner11}.\footnote{See also \citep{BarrettKentPironio06} and references therein for the use of chained Bell inequalities.}
A proof that is more adapted to the formulation used here can be found in the proofs of Theorems 10.4 and 10.7 in \citep{Leifer14}. 
For completeness a proof is given in \autoref{EqProbSec}, which will also serve as a stepping stone for the proof of a generalized version of the theorem below.
 
Although the equiprobability theorem is still far removed from \autoref{claim1}, it deserves to be noted that the result is an improvement on earlier work.
\Citet{Stairs83} showed that for any pair in a maximally entangled $d$-level systems with $d\geq3$, there is no value definite ontic model that satisfies parameter independence.\footnote{This result later became known as the free will theorem by \citet{ConwayKochen06,ConwayKochen09}.}
In the present formulation, value definiteness can be understood as the assumption that the $\lambda$-probabilities are 0,1-valued. 
The work by Colbeck and Renner improves on this result on two accounts: the result also holds for the case $d=2$ and, not only must ontic models be probabilistic, they must even follow the quantum probability rule.

Like the proof of Stairs' theorem, the proof of the equiprobability theorem relies heavily on the special properties of maximally entangled states.
Because of the symmetry in the maximally entangled state, it assigns equal probabilities to all outcomes for a non-disperse observable.
This implies that it suffices to show that all $\lambda$-probabilities should be equal.
Thus no direct relation to the Born rule needs to be established.
The proof that all the $\lambda$-probabilities are equal in turn relies heavily on the perfect correlations of the maximally entangled state.

The extension of the equiprobability theorem still relies on perfect correlations, but relaxes the symmetry property.
Specifically, one makes the generalization to arbitrary entangled states, but does so by restricting to perfectly correlated measurements, i.e., measurements whose corresponding operator is diagonal in a Schmidt basis.
To obtain this result, one further needs to extend the framework of ontic models.
This is done in the next subsection.

\subsection{The extended equiprobability theorem}

The scenario of the extended equiprobability theorem is again a pair of $d$-level quantum systems, but now we consider an arbitrary state $\psi_d^S$.
Instead of arbitrary local measurements, one now only considers measurements whose corresponding operator is diagonal with respect to some orthonormal basis  $(e_i)_{i=1}^d$ for $\mathbb{C}^d$ such that
\begin{equation}\label{psidsdef}
	\psi_d^S=\sum_{i=1}^dc_ie_i\otimes e_i,
\end{equation}
for some numbers $c_i$.\footnote{Note that in the case where $\psi^S_d$ is a maximally entangled state, we again obtain the set of all local measurements. On the other hand, if $\psi_d^S=e\otimes e$, then the only local measurements considered are those whose corresponding operator has $e$ as an eigenstate.}
The strategy to prove that the $\lambda$-probabilities for these measurements should equal the quantum probabilities relies on coupling the pair of systems to a pair of $D$-level systems.
The combined system is then brought to a state that is approximately maximally entangled and then the equiprobability theorem is applied so probabilities for local measurements should then approximately equal $\tfrac{1}{D}$.
These probabilities are then related to the $\lambda$-probabilities for the initial pair of systems, which can then be approximated to have the value $\tfrac{n_i}{D}\approx\left|c_i\right|^2$.

To deal with the notion of coupling one system to another, one needs a way to translate the ontic models for two (or more) quantum systems to the ontic model for the joint quantum system.
Here enters the idea for what will be called a \emph{complete ontic model}.\footnote{This definition is based on Leifer's proposal for how to deal with appending ancillas in ontic models \citep[\S8.2]{Leifer14}, which is also relevant for his proof of \autoref{claim2}.}
\begin{definition}
A \emph{complete ontic model} for quantum mechanics is a collection of ontic models, with one ontic model for every finite-dimensional quantum system.
For every pair of quantum systems with Hilbert spaces $\h_1$ and $\h_2$ and for any system with Hilbert space $\h_1\otimes\h_2$ there is for each preparation $P_2$ of the second system a Markov kernel $\gamma_{P_2}$ from the ontic model $(\Lambda_1,\Sigma_1)$ for the first system to the ontic model $(\Lambda_{12},\Sigma_{12})$ for the joint system such that for every preparation $P_1$ of the first system the probability measure $\mu_{P_1P_2}$ defined as
\begin{equation}
	\mu_{P_1P_2}(\Delta)\ldef\int\gamma_{P_2}(\Delta|\lambda)\dee\mu_{P_1}(\lambda)
\end{equation}
corresponds to some preparation $P_1P_2$ of the joint system such that if $P_1$ is represented by the quantum state $\rho_1$ and $P_2$ by $\rho_2$, then $P_1P_2$ is represented by $\rho_1\otimes\rho_2$.

A complete ontic model is said to satisfy \emph{parameter independence} if all its ontic models satisfy parameter independence.
\end{definition}

One can think of the Markov kernel $\gamma_{P_2}$ as modeling the act of appending an ancilla prepared according to $P_2$ to the initial system.
It provides a translation from the ontic model for the initial system to the ontic model for the joint system.
It seems natural to demand that the description of the initial system in the ontic model for the joint system, now considered as a subsystem, should be at least as rich as its description in the initial ontic model.
The assumption of ancilla independence captures this intuition.

\begin{definition}
A complete ontic model is said to satisfy \emph{ancilla independence} when $\lambda$-probabilities for an individual system arise as averages of $\lambda$-probabilities for local measurements on a joint system.
Specifically, for every response function $\peetje_{M_A}$ in the initial ontic model there is a response function $\peetje_{M_{A\otimes\een}}$ in the ontic model for the joint system such that  for every preparation $P_2$ for a second system
\begin{equation}\label{AnApp}
	\peetje_{M_A}(a|\lambda)=\int \peetje_{M_{A\otimes\een}}(a|\tilde{\lambda})\gamma_{P_2}(\dee\tilde{\lambda}|\lambda)
\end{equation}
for all possible measurement outcomes $a$ and all ontic states $\lambda\in\Lambda_1$.
\end{definition}

At first sight the notion of ancilla independence may be reminiscent of preparation independence as adopted in the PBR theorem, since the validity of \eqref{AnApp} is independent of the choice of $P_2$.
Note though that unlike preparation independence, which relies on the Cartesian product assumption \citep[p. 100]{Leifer14}, ancilla independence does not have direct implications for the structure of the set of ontic states for the joint system and how these ontic states relate to the ontic states for the individual systems.
It seems to me that ancilla independence is a reasonable assumption.
Either way, it is needed to prove the extended equiprobability theorem.

\begin{theorem}[Extended equiprobability theorem]\label{extequipthm}
For any complete ontic model that satisfies ancilla independence and parameter independence let an ontic model for a pair of $d$-level systems with $d\geq2$ be given and let $\psi_d^S$ be any quantum state for the pair.
Then the ontic model must be trivial with respect to the set of preparations that are represented by $\psi_d^S$ and the set of local measurements whose corresponding operator is diagonal in an orthonormal basis $(e_i)_{i=1}^d$ for which \eqref{psidsdef} holds for some numbers $c_i$.
\end{theorem}

A proof for this theorem is presented in \autoref{appendix2}.
As far as I know it is the first fully rigorous proof of this type of theorem.
Of course, the theorem as presented here is strictly speaking new simply because it is formulated in the language of ontic models. 
It also does not map directly onto earlier theorems and proofs because the strategy for purported proofs for Claim 1 have variations.
Nevertheless, some general remarks can be made.

The criticism of \citet{Landsman15} pertaining to earlier proofs by Colbeck and Renner can be traced to the fact that the proof makes use of three limiting operations that are intertwined and that Colbeck and Renner appear to rely on special mathematical properties of the formalism of quantum mechanics to deal with these limiting operations.
Landsman partly deals with these issues by resorting to strong assumptions, which he takes to be part of the analysis of Colbeck and Renner, but only now are being made explicit.
For example, one key strategy of the proof is to couple the initial pair of systems to another pair of $d^\epsilon$-level systems.
This second pair is then evolved to a state that is close to a maximally entangled state.
And although the equiprobability theorem only applies to exactly maximally entangled states, it is nevertheless applied.
Landsman uses here the assumption of continuity of probabilities (based on the idea that the Born rule is continuous) to bridge the gap.

\citet{Leegwater16} at first sight seems to avoid this issue by working in the limit where the $d^\epsilon$-level systems can evolve to a maximally entangled state.
This evolution in turn relies on coupling the two pairs of systems to a third pair of $D$-level systems.
As $D\to\infty$, the entangled state for the $d^\epsilon$-level systems tends towards the maximally entangled state.
But it is not clear whether working in the limit $D=\infty$ is legitimate since the proof of the equiprobability theorem itself relies on a limiting operation: the chained Bell inequalities (this is the second limiting operation).
The final limiting operation is that of approximating the quantum probabilities $\left|c_i\right|^2$ with rational numbers $\tfrac{n_i}{d^\epsilon}$ where the approximation becomes better as $d^\epsilon\to\infty$.
Here again Leegwater is not very clear about whether the limit $d^\epsilon\to\infty$ can directly be applied to a result that itself relies on a limiting operation.
These issues are avoided in the proof in \autoref{appendix2}, where all limiting operations are postponed to the final steps, and where it also becomes clear that in fact these operations are intertwined.

\section{Completeness for individual systems}\label{qubitsec}
The idea of a completeness proof for individual quantum systems seems peculiar given that non-trivial ontic models for arbitrary $d$-level quantum systems have been around since the work of \citet{Bell66} and \citet{Gudder70}.
Trivially, these models may be assumed to satisfy parameter independence since there is no second system in play with which it could interact.
\emph{So strictly speaking these models are counterexamples to Claim 1.}

The lack of a description of interaction in these models may be considered to be a serious deficiency.
In fact, in light of the extended equiprobability theorem, it seems that Claim 1 should at least be reformulated to the weaker
\begin{claim}\label{claim1p}
In any complete ontic model that satisfies ancilla independence and parameter independence every ontic model is $\psi$-complete.
\end{claim}
In principle, one can imagine that under certain minimal assumptions on interactions, parameter independence becomes applicable and ontic models necessarily have to become trivial.
This seems at least to be the aim of Colbeck and Renner, and so it is useful to investigate how their strategy is supposed to work.

Consider an individual quantum system with Hilbert space $\h_1=\mathbb{C}^d$ and a preparation corresponding to the pure state
\begin{equation}
	\psi_1=\sum_{i=1}^dc_ie_i,
\end{equation}
for some orthonormal basis $\{e_1,\ldots,e_d\}$.
The focus is on a possible measurement represented by a complete self-adjoint operator $A=\sum_{i=1}^da_i[e_i]$.
It is straightforward to devise an ontic model for this system that is non-trivial with respect to this measurement for this state.
That is, a model in which 
\begin{equation}\label{nontrivqubit}
	\int\left|\peetje_A(a_i|\lambda)-\left|c_i\right|^2\right|\dee\mu_{\psi_1}(\lambda)\neq0. 
\end{equation}
This is in striking contrast to the consequence of the extended equiprobability theorem which says that for any ontic model for a pair of such systems (that is part of an appropriate complete ontic model) one has
\begin{equation}\label{equiprobeq}
	\int\left|\peetje_{A\otimes\een}(a_i|\lambda)-\left|c_i\right|^2\right|\dee\mu_{\psi_d^S}(\lambda)=0
\end{equation}
with $\psi_d^S$ as in \eqref{psidsdef}.
Could it be possible to argue from this that any ontic model for the individual system in which \eqref{nontrivqubit} is the case should be rejected?

The step needed to make the connection is most explicit in the work of \citet[\S8]{Leegwater16}.\footnote{In \citep{ColbeckRenner11} the relevant step connects to assumption QMb in the supplementary material and in \citep{ColbeckRenner15Comp} it is discussed in section 6.}
Instead of looking at $A\otimes\een$ the focus is on $\een\otimes B$ where $B=\sum_{i=1}^db_i[e_i]$.
By \autoref{equipthm} the relevant analogue of \eqref{equiprobeq} also holds for measurements represented by this operator.
The step is then that since ``by definition''\footnote{The notation of the equations \eqref{QMeq} and \eqref{MuddyEq1} has been adjusted to fit the notation in this paper.}
\begin{equation}\label{QMeq}
	\pee(a_i|A,\psi_1)=\pee(b_i|\een\otimes B,\psi_d^S)
\end{equation}
``the same relation holds when considering $\lambda$-probabilities:''
\begin{equation}\label{MuddyEq1}
	\peetje^{\psi_1}_A(a_i|\lambda)=\peetje^{\psi_d^S}_{\een\otimes B}(b_i|\lambda).
\end{equation}

Presumably, \eqref{MuddyEq1} is supposed to imply that \eqref{nontrivqubit} cannot hold in any ontic model that is part of a complete ontic model that satisfies ancilla independence and parameter independence.
Whether this is the case of course depends on what \eqref{MuddyEq1} exactly expresses.
This is not entirely trivial as a formal definition of the expression is lacking.

As a first step in fleshing out what \eqref{MuddyEq1} means, assume that the inference from \eqref{QMeq} to \eqref{MuddyEq1} is valid if and only if the inference from
\begin{equation}\label{QMeq2}
	\pee(a_i|A,\psi_1)=\pee(a_i|A\otimes\een,\psi_d^S)
\end{equation}
to
\begin{equation}\label{MuddyEq}
	\peetje^{\psi_1}_A(a_i|\lambda)=\peetje^{\psi_d^S}_{A\otimes\een}(a_i|\lambda)
\end{equation}
is valid.
This is reasonable since if one inference holds, then, by the extended equiprobability theorem, the other holds as well.

In quantum mechanics \eqref{QMeq2} holds as a consequence of the mathematical structure of the theory and how it deals with composing joint systems out of individual systems. 
But it is important to note that the equation expresses a numerical equivalence of two probabilities that are defined in separate models.
The objects $A$ and $\psi_1$ belong to the quantum model for the individual system while $A\otimes\een$ and $\psi_d^S$ belong to the pair model.
The fact that the same symbol $\pee$ is used on both sides of the equation does not mean that the equation establishes the equality of two values a single function takes on for two distinct arguments; the two functions are distinct.

To have a similar link between $\lambda$-probabilities of bipartite systems and $\lambda$-probabilities of their parts, one has to suppose that the ontic models for all relevant systems are part of (something like) a complete ontic model.
But even then it is not clear what it means for the same $\lambda$ to appear on both sides of the equation, just as it wouldn't make sense to have the same quantum state on both sides of the equation in \eqref{QMeq2}.
These are just some preliminary considerations for tackling the real issue at hand, namely, that the objects on both sides of \eqref{MuddyEq} also are not well-defined.
As these issues are intertwined I will deal with them simultaneously. 

A small part of the problem stems from trying to analyze \eqref{MuddyEq} in the language of ontic models, while Colbeck, Renner and Leegwater take a different approach.
On their approach, both $\psi_1$ and $\lambda$ are possible values for random variables, as is the measurement setting and the measurement outcome.
Then $\peetje^{\psi_1}_A(a_i|\lambda)$ expresses the probability of obtaining the outcome $a_i$ conditional on the quantum state being $\psi_1$ and the measurement being $A$ and something else being $\lambda$.\footnote{It may be noted here that even within their framework, the $\lambda$-probability with $\psi$ sticked to it is only well-defined if the conjunction of $A$, $\psi_1$ and $\lambda$ has non-zero probability.}
It may be possible that they intend the random variable that determines the quantum state to not just range over states for a specific Hilbert space, but also across different Hilbert spaces.
Then both $\psi_1$ and $\psi_d^S$ are possible values and so there is a well-defined probability that the system will be a single $d$-level system and a well-defined probability that it will be a pair, determined by the single function $\peetje$.
But this idea will be avoided here.

Another important distinction is that Colbeck, Renner and Leegwater do not presuppose that $\lambda$ provides all relevant information concerning the system and so ``adding'' $\psi$ to $\lambda$ may give more information about the possible outcomes.  
This explains the occurrence of quantum states in \eqref{MuddyEq}.
Informally speaking, the quantum state may also be taken to give relevant information within the use of an ontic model.
As it specifies the preparation of the system, it may indicate what behavior of $\lambda$ may be considered to be ``typical''.
So the notation in \eqref{MuddyEq} may be interpreted as providing a convenient way of specifying a constraint on the possible $\lambda$'s.
A way to make this precise, is to assume (like \citet{Landsman15}) that $\lambda$-probabilities with quantum states attached to them express equations that hold $\mu_{\psi_1}$-almost surely, i.e.,
\begin{equation}
\begin{gathered}
	\peetje_A^{\psi_1}(a_i|\lambda)=f(\lambda)\\
	\iff\\
	\int\left|\peetje_A(a_i|\lambda)-f(\lambda)\right|\dee\mu_{\psi_1}(\lambda)=0.
\end{gathered}
\end{equation}

This helps in clarifying how to interpret both sides of \eqref{MuddyEq}.
It does not help yet in interpreting what equality between the two means.
A problem with \eqref{MuddyEq} is that it refers to two distinct quantum states (never mind that they also belong to distinct Hilbert spaces).
\citet{Landsman15} doubles down on the almost surely interpretation here and proposes the definition
\begin{equation}\label{LandsmanCrit}
\begin{gathered}
	\peetje^{\psi_1}_A(a_i|\lambda)=\peetje^{\psi_d^S}_{A\otimes\een}(a_i|\lambda)\\
	\iff\\
	\forall f:\peetje_A^{\psi_1}(a_i|\lambda)=f(\lambda)\text{ iff }
	\peetje^{\psi_d^S}_{A\otimes\een}(a_i|\lambda)=f(\lambda).
\end{gathered}
\end{equation} 
A first problem here is that, in the current setting, this is still not well-defined as the $\lambda$'s belong to distinct ontic models and so it is not clear what kind of function $f$ could even have the proper domain.
But even if that issue can be resolved, \eqref{LandsmanCrit} is an unreasonably strong assumption, as it implies that if \eqref{MuddyEq} holds, then $p_A(a_i|\lambda)=p_{A\otimes\een}(a_i|\lambda)$ both $\mu_{\psi_1}$-almost surely and $\mu_{\psi_d^S}$-almost surely.
Moreover, distributions for $\mu_{\psi_1}$ and $\mu_{\psi_d^S}$ completely overlap on the region of $\Lambda$ where $\peetje_A(a_i|\lambda)$ is non-zero.\footnote{This may be seen as follows. 
The first claim follows from evaluating \eqref{LandsmanCrit} for the choice $f(\lambda)=p_A(a_i|\lambda)$. 
The second claim follows with a proof from contradiction. 
If there exists a $\Delta\subset\Lambda$ such that $\mu_{\psi_d^S}(\Delta)>0=\mu_{\psi_1}(\Delta)$ and on which $\peetje_{A\otimes\een}(a_i|\lambda)$ is non-zero, then for any $f$ such that $\peetje_A^{\psi_1}(a_i|\lambda)=f(\lambda)$ and $\peetje_{A\otimes\een}^{\psi_d^S}(a_i|\lambda)=f(\lambda)$ one can define an $f'$ that differs from $f$ only on $\Delta$ in such a way that $\peetje_{A\otimes\een}^{\psi_d^S}(a_i|\lambda)=f'(\lambda)$ no longer holds. 
So \eqref{MuddyEq} would fail.}
This is in strong tension with results like the BCLM theorem \citep{BCLM14}, which demonstrate that overlaps should become arbitrarily small as $d$ increases.

A solution to the problem can be found in the work of \citet{Leegwater16}, and it overlaps with the solution to the double use of $\lambda$ for physically distinct systems.
In the discussion of some of the notation used that is similar to that occurring in \eqref{MuddyEq}, Leegwater notes:
\begin{quote}
here $\lambda$ still refers to the variable assigned to system $A$ when it was in the state $\ket{\psi}_A$ [\ldots]
$\lambda$ always refers to the original system $A$, and there is only one measure $\mu(\lambda)$ that is considered. [p.21]
\end{quote}
The measure $\mu$ refers to a probability distribution over $\Lambda$ that may depend on $\ket{\psi}_A$ as well as other factors.
Given that there is only one measure considered, which is related to a particular quantum state, all other quantum states should be understood as being arrived at after interactions with a system prepared according to $\ket{\psi}_A$.

Translating this to the present discussion, it implies that $\peetje^{\psi_d^S}_{A\otimes \een}(a_i|\lambda)$ refers to a single $d$-level system that was initially prepared in the quantum state $\psi_1$ and having ontic state $\lambda$, then was coupled to a second $d$-level system with unknown state and then a transformation on the joint system was performed yielding the quantum state $\psi_d^S$.
Within a complete ontic model this process can be modeled.
For sake of definiteness, assume that the second system was prepared according to the quantum state $\phi$.
Appending this system to the first one is then modeled by a Markov kernel $\gamma_\phi$ from the initial ontic model to the model for the joint system.
Now let $U$ be a unitary operator such that $U(\psi_1\otimes\phi)=\psi_d^S$ and let $\gamma_U$ be a Markov kernel that models it.
A proper reformulation of \eqref{MuddyEq} is then that\footnote{Note that both integrals are over the ontic states for the joint system, while $\lambda$ is an ontic state for the initial system.}
\begin{equation}\label{NoMoreMud}
	\peetje_A(a_i|\lambda)=\iint\peetje_{A\otimes\een}(a_i|\tilde{\lambda})\gamma_U(\dee\tilde{\lambda}|\lambda')\gamma_{\phi}(\dee\lambda'|\lambda)
\end{equation}
holds $\mu_{\psi_1}$-almost surely.

This reconstruction seems to be on the right track.
If the inference from \eqref{QMeq2} to \eqref{NoMoreMud} is valid, then the extended equiprobability theorem can be used to prove the completeness claim.
Since $\mu_{T_U\circ T_\phi\circ P_{\psi_1}}$ corresponds to a preparation that can be represented by the state $\psi_d^S$, it follows that $\peetje_{A\otimes\een}(a_i|\tilde{\lambda})$ is almost surely equal to the quantum probability $\left|c_i\right|^2$ with respect to this measure.
It then follows from \eqref{NoMoreMud} that $\peetje_A(a_i|\lambda)$ is equal to the quantum probability $\mu_{\psi_1}$-almost surely.
The remaining question is then why the inference from \eqref{QMeq2} to \eqref{NoMoreMud} should be valid. 
In the next section I argue that no uncontroversial arguments are available.

\section{Unitary processes and measurements}\label{critsec}
\subsection{A valid alternative to Claim 1}
It is not obvious from the works of Colbeck, Renner and Leegwater what should be taken to be the main argument for the inference from \eqref{QMeq2} to \eqref{NoMoreMud}.
To get a handle on what it takes for \eqref{NoMoreMud} to hold, I show that it can be derived with the help of two assumptions.
The first assumption is ancilla independence as introduced in \autoref{partialsec}.
The second I call unitary invariance.
In spirit it is similar to the assumption of unitary invariance adopted by \citet{Landsman15}, but it avoids problems similar to those related to \eqref{LandsmanCrit}.
\begin{definition}
An ontic model is said to satisfy \emph{unitary invariance} if for every $\rho,U,A,a$
\begin{equation}\label{FatU}
\begin{gathered}
	\mathrm{Tr}\left(\rho [a]_A\right)=\mathrm{Tr}\left(U\rho U^*[a]_A\right)\\
	\implies\\
	\peetje_{M_A}(a|\lambda)=\int \peetje_{M_A}(a|\tilde{\lambda})\gamma_{T_U}(\dee\tilde{\lambda}|\lambda)~
	\mu_{P_\rho}\text{-a.s.}
\end{gathered}
\end{equation}
\end{definition}
This definition states that, given some preparation represented by $\rho$, whenever some transformation does not alter the probabilities for a particular measurement and outcome according to quantum mechanics, then also the $\lambda$-probabilities are not altered by the transformation (almost surely with respect to the preparation).\footnote{It is not to be confused with the notion of transformation noncontextuality \citep{Spekkens05}, which states that if two transformation affect all the operational probabilities in the same way, then they should also affect all the $\lambda$-probabilities in the same way.}

With these assumptions \eqref{NoMoreMud} can indeed be derived:\footnote{Here $U$ is as before, i.e., $U(\psi_1\otimes\phi)=\psi^D_s$.  Ancilla independence is used in the first step, unitary invariance is used in the final step.}
\begin{multline}\label{Derivation}
	\int\left|
	\peetje_A(a_i|\lambda)\vphantom{\iint}
	-\int \peetje_{A\otimes\een}(a_i|\tilde{\lambda})\gamma_{U\circ\phi}(\dee\tilde{\lambda}|\lambda)
	\right|\dee\mu_{\psi_1}(\lambda)\\
	\begin{aligned}
	={}&
	\int\left|
	\int\peetje_{A\otimes\een}(a_i|\lambda')\gamma_\phi(\dee\lambda'|\lambda)
	-
	\int \peetje_{A\otimes\een}(a_i|\tilde{\lambda})\gamma_{U\circ\phi}(\dee\tilde{\lambda}|\lambda)
	\right|\dee\mu_{\psi_1}(\lambda)\\
	\leq{}&
	\iint\left|
	\peetje_{A\otimes\een}(a_i|\lambda')
	\vphantom{\int}
	-
	\int \peetje_{A\otimes\een}(a_i|\tilde{\lambda})\gamma_U(\dee\tilde{\lambda}|\lambda')
	\right|\gamma_\phi(\dee\lambda'|\lambda)\dee\mu_{\psi_1}(\lambda)\\
	={}&
	\int\left|
	\peetje_{A\otimes\een}(a_i|\lambda)
	\vphantom{\int}
	-
	\int \peetje_{A\otimes\een}(a_i|\tilde{\lambda})\gamma_U(\dee\tilde{\lambda}|\lambda)
	\right|\dee\mu_{\psi_1\otimes\phi}(\lambda)
	=0.
	\end{aligned}
\end{multline}
We thus obtain a proof for the following theorem (which is to be contrasted to \autoref{claim1} and \autoref{claim1p}):
\begin{theorem}
In any complete ontic model that satisfies ancilla independence, unitary invariance and parameter independence every ontic model is $\psi$-complete. 
\end{theorem}

The value of this theorem depends on whether unitary invariance is a reasonable assumption.\footnote{Assuming for convenience that the reader does think setting independence and parameter independence are reasonable assumptions.}
However, it is not trivial to come up with a motivation for unitary invariance.
This is problematic, as this assumption on its own has immediate consequences that are quite crucially related to the notion of completeness.
Note that completeness boils down to the idea that response functions are dispersion free under preparations of pure quantum states.
Now consider the example where $\expv{\psi}{[a]_A}{\psi}=\expv{U\psi}{[a]_A}{U\psi}$.
Unitary invariance then has the following consequence for the dispersion of any corresponding response function:
\begin{equation}
\begin{split}
	\mathrm{Var}_\psi(\peetje_A(a|\:.\:))
	=&
	\int \peetje_A(a|\lambda)^2\dee\mu_\psi(\lambda)-\left(\int \peetje_A(a|\lambda)\dee\mu_\psi(\lambda)\right)^2\\
	=&
	\int \left(\int \peetje_A(a|\tilde{\lambda})\gamma_U(\dee\tilde{\lambda}|\lambda)\right)^2\dee\mu_\psi(\lambda)
	-
	\left(\iint \peetje_A(a|\tilde{\lambda})\gamma_U(\tilde{\lambda}|\lambda)\dee\mu_\psi(\lambda)\right)^2\\
	\leq&
	\iint \peetje_A(a|\tilde{\lambda})^2\gamma_U(\dee\tilde{\lambda}|\lambda)\dee\mu_\psi(\lambda)
	-
	\left(\iint \peetje_A(a|\tilde{\lambda})\gamma_U(\tilde{\lambda}|\lambda)\dee\mu_\psi(\lambda)\right)^2\\
	=&
	\mathrm{Var}_{U\psi}(\peetje_A(a|\:.\:)).
\end{split}
\end{equation}

It follows that if one can show that for some quantum state $U\psi$ the response function is dispersion free for a particular measurement, then the response function must also be dispersion free for any other quantum state $\psi$ whenever the operational probabilities for the measurement are equal.
Since this gets at the heart of the Colbeck-Renner completeness claim, this implicit assumption should have at least been explicitly stated and preferably be well-mo\-ti\-vated. 
Neither appears to be the case though.

\subsection{An appeal to measurement theory}

The crucial inferences that implicitly rely on unitary invariance occur in the setting where measurement procedures are considered  \citep[p.4-5]{ColbeckRenner11}, \citep[p.515]{ColbeckRenner15Comp} and \citep[\S8]{Leegwater16}. 
Possibly then, an appeal to measurement procedures can help motivate unitary invariance or some weaker version thereof.
Arguably, an appeal to what actually happens in a measurement process is not a very elegant strategy in a theory that is infamous for its measurement problem.
One should be careful not to adopt assumptions for the ontic model that don't sit well within quantum mechanics itself.

Ontic models are usually not designed to resolve the measurement problem, as the Beltrametti-Bugajski model nicely illustrates.
More generally, in any ontic model probabilities for outcomes of measurements are well-defined without any mentioning of how the relevant system would interact with a measurement apparatus.
Ideally, this peculiarity is to be solved in a complete ontic model. 
The $\lambda$-probability $\peetje_M(a|\lambda)$ then encodes the probability with which the system, upon interaction with an appropriate measurement apparatus, would evolve towards a joint state in which the measurement apparatus can be taken to be in a well-defined pointer state displaying the outcome of the measurement.
Explicitly, if $\Delta^M_a$ denotes the set of ontic states of the joint system in which the apparatus displays the outcome $a$, then ideally there is a Markov kernel $\gamma_M$ modeling the interaction between system and apparatus such that
\begin{equation}\label{MeasureOntic}
	\peetje_M(a|\lambda)=\gamma_M\left(\Delta^M_a\middle|\lambda\right).
\end{equation}
A motivation for unitary invariance in the context of measurements could then boil down to an assumption concerning $\gamma_M$.

That an appeal to measurement processes can be relevant for unitary invariance can be seen from a simple example.
According to quantum theory, instead of directly measuring $A$ on the initial system in the state $\psi_1$, one may equivalently first entangle it with a second system to obtain the joint state $\psi_d^S$ and measure instead $B$ on the second system.
The unitary process here is non-disturbing with respect to the observable measured.
Unitary invariance comes down to the idea that this is also always a legitimate procedure for the ontic model.
However, the inference from operational non-disturbance to non-disturbance on the level of ontic models is known to be problematic \citep{Maroney17}.

An appeal to measurement processes to motivate the \emph{general} validity of unitary invariance thus seems implausible.
However, for the proof of Theorem 3 it is only used in certain special circumstances.
Given a system and a measurement $A$ and outcome $a_i$, it needs to be shown that the $\lambda$-probability $\peetje_A(a_i|\lambda)$ equals the quantum probability almost surely (with respect to the preparation of the quantum state).
What happens on the level of the ontic model is that during the measurement procedure the system evolves according to $\gamma_M$.
Now, if one could argue that $\gamma_M$ actually involves a unitary process by which the system plus apparatus evolve to the quantum state $\psi^S_d$ after which the outcome is read of from the apparatus, then maybe the extended equiprobability theorem can be applied and one can conclude that indeed $\peetje_A(a_i|\lambda)$ equals the quantum probability almost surely.

The transition from $\psi_1$ to $\psi_d^S$ is of course well-known as part of the von Neumann measurement scheme.
The usual problem is then that the state $\psi_d^S$, on its own, does not signify a situation in which a measurement outcome is obtained.
Single world unitary quantum mechanics without hidden variables is not compatible with the assumption that measurements have definite outcomes.

The question is how the measurement process $\gamma_M$ relates to the von Neumann-type process $\gamma_U\circ\gamma_\phi$.
And this may of course depend on how the complete ontic model is supposed to solve the measurement problem.\footnote{It is not entirely clear if Colbeck and Renner expect the alternative theory to solve the measurement problem. Doing so would seem unreasonable as it sets a higher standard for the alternative theory than adopted for quantum mechanics. But not doing so makes the appeal to measurement processes quite awkward. Regardless, the following arguments do not rely on whether or not $\gamma_M$ actually satisfies \eqref{MeasureOntic}. It just denotes the measurement process in the ontic model, whatever it may be.}
In a spontaneous collapse theory, the macroscopic superposition $\psi_d^S$ is extremely unlikely to obtain and $\gamma_M$ contains a collapse way before the unitary evolution $U$ is completed.
Consequently, an appeal to measurements is irrelevant for the question whether \autoref{claim1p} applies to such theories and a general assumption like unitary invariance is required instead.

To be fair, in the first paper \citet{ColbeckRenner11} did assume that ``all processes within quantum theory can be considered as unitary evolutions''.
Although this explicit statement is dropped in any of the following papers, a charitable reading is that the completeness claim was never meant to apply to collapse theories.\footnote{It may be noted that by such a step it must be conceded that the theorem of Colbeck and Renner is not on an equal footing with results like Bell's theorem, the Kochen-Specker theorem or the PBR theorem.}
But even if it is accepted that $\gamma_M$ should correspond to a unitary process, it is not obvious that it should be operationally equivalent to $\gamma_U\circ\gamma_\phi$ (i.e., can be represented by the same unitary operator).
As an alternative process, consider some unitary operation $U'$ that evolves the total system to a state of the form $\sum_{i=1}^dc_i\chi\otimes e_i$ for some arbitrary pure state $\chi$.
This represents as much a process in which the final state describes an apparatus that displays the outcome $a_i$ with probability $|c_i|^2$ as does the operation $U$.
And there is no principled reason that in ontic models the measurement process $\gamma_M$ can not be of this form.\footnote{I picked up this idea from \citep{Wallace19}, which contains more insightful discussion of what ``orthodox quantum mechanics'' should be taken to mean.} 
And if it is, the extended equiprobability theorem cannot be applied.\footnote{One may be inclined to insist that within quantum mechanics such a unitary process $U'$ does not represent a measurement. But that is beside the point, as we are considering here possible measurement processes within ontic models for which the only criterion is that the Born rule probabilities are reproduced.}


Thus we obtain again a further restriction on the set of complete ontic models to which the theorem is to be applied.
The models must be such that all processes correspond to unitary processes and measurement processes like $\gamma_M$ necessarily lead to entangled states with perfect correlations.
To see if at least this claim is valid, it is worthwhile to quickly recap what the argument for it has become.
The procedure to measure $A$ on a system, involves hooking it up to a measurement device, evolving the joint system to $\psi_d^S$ and then reading of the value of $B$ from the device.
Reading of this value constitutes a measurement of $B$ on the measurement device.
Because in a no-collapse theory with von Neumann measurement processes this is roughly what it means to measure $A$, the outcome statistics for $B$ on the joint system should equal the outcome statistics for $A$.
But, if in addition the criteria for the extended equiprobability theorem are in place, then the outcome statistics for $B$ should be equal to the quantum statistics and so also the statistics for $A$ should equal the quantum statistics.

Applicability of the extended equiprobability theorem requires parameter independence.
Thus it has to be assumed that it is possible to postpone the $B$-measurement until system and apparatus are again spatially separated.
Moreover, at this stage it should also be possible to refrain from performing a $B$-measurement, and instead perform any other measurement on the apparatus and, in addition, to perform any other measurement on the system as well.
If not, the chained Bell-inequalities upon which the proof for \autoref{extequipthm} rests cannot come off the ground.

It first may be acknowledged that in unitary quantum mechanics it should be possible in principle to entangle a system with an apparatus, then spatially separate them and in turn perform arbitrary measurements on both systems. 
It is of course not clear how we are to understand a measurement on an apparatus that is not a measurement in the pointer basis, just like it is difficult to measure a cat in a basis that contains superpositions of life and death states.
But we should not let such psychological worries affect our judgments.

Thus there should be a process $\gamma_U\circ\gamma_\phi$ in the ontic model that establishes an entangled state for the joint system such that after which any measurement can be performed on each of the components of the system.
Is it reasonable to think of $\gamma_M$ as such a process?
Here it is good to remind ourselves that in an ontic model the physical quantities may correspond only to dispositional properties.\footnote{This separates ontic models from more traditional hidden variable models, in which value definiteness of physical quantities is usually taken to be among the criteria.}
The process $\gamma_M$ may be seen to be special in the sense that it brings about such a dispositional property.\footnote{This is not necessarily re-introducing the measurement problem. What is special about $\gamma_M$ is that it brings about value definiteness, and it is reasonable to take that as a necessary requirement for the processes that we call measurements.}

Now once a certain dispositional property is brought about, it is not per se meaningful to say that dispositional properties before the measurement encoded by other response functions are still there. 
What we can have is an effective collapse.
This is akin to an Everettian multiverse view. 
In that scenario, once we have done a measurement we find ourselves in a branch that is ``cut off'' from other branches. 
Whatever future processes we may apply to the system after branching, these processes apply only to the branch we find ourselves on.
Future measurements do no longer bring about dispositional properties of the system prior to the first measurement, but only of a branched off version of the system.
In a single world universe, the idea is that once a context of measurement has been brought about, we have cut ourselves off from other possible contexts.\footnote{One might worry that this idea leads to possible conflicts with unitary quantum mechanics (on some specific reading of what unitary quantum mechanics actually is). But all I am using here is the, I think, uncontroversial idea that, if I read the outcome on a measurement device twice, then with probability one the second reading corresponds to the first reading, even if the quantum state has not collapsed to the corresponding eigenstate. This is just what effective collapse means.}
This idea is of course very Bohrian in spirit.
In fact, Bohr's response to Einstein, Podolsky and Rosen, when taken out of context, seems strangely apt to qualify the issue regarding the stage when $\psi_d^S$ obtains, if one changes just one word, and it is perfectly fit to serve as the main conclusion of this paper:
\begin{quote}
But even at this stage there is essentially the question of \emph{an influence on the very conditions which define the possible types of predictions regarding the future behavior of the system}.
Since these conditions constitute an inherent element of the description of any phenomenon to which the term ``physical reality'' can properly be attached, we see that the argumentation of the mentioned authors does not justify their conclusion that the quantum-mechanical description is essentially [complete]. \citep[p.700]{Bohr35}
\end{quote}

\section*{Acknowledgments}
I would like to thank Guido Bacciagaluppi, Klaas Landsman and Gijs Leegwater for fruitful and intense discussions on the work of Colbeck and Renner. 
Without them, this paper never would have seen the light of day.
I further would like to thank the people at the BSPS and ILMPS conferences, where I was given the opportunity to present the main ideas in this paper.
And also two anonymous referees for some serious comments that have attributed to major improvements of the paper.
This research was funded by the Netherlands Organisation for Scientific Research (NWO), Veni Project No. 275-20-070.

\begin{appendix}
\section{Proof of the equiprobability theorem}\label{EqProbSec}
In this appendix the equiprobability theorem is proven.
For convenience the theorem is formulated again here.
\begin{theorem*}
Consider the quantum system of a pair of $d$-level systems.
Any ontic model for this system that satisfies parameter independence is trivial w.r.t. $\mathcal{M}_{\mathrm{LOC}}$ and all preparations represented by the state
\begin{equation*}
	\psi_d=\sum_{i=1}^{d}\frac{1}{\sqrt{d}}e_i\otimes e_i.
\end{equation*}
\end{theorem*}

\begin{proof}
Throughout, assume that $\mu_{\psi_d}$ is an arbitrary, but fixed, probability measure corresponding to a preparation of $\psi_d$.
First consider a complete local measurement $M_A$ with possible outcomes $a_1,a_2,\ldots,a_d$ corresponding to the operator $A\otimes\een$.
For the state $\psi_d$, the probability for each possible outcome is $\tfrac{1}{d}$.
To show that the $\lambda$-probabilities are identical, it suffices to show that the $\lambda$-probabilities for all possible outcomes are equal:
\begin{equation}
\begin{split}
	\int\left|\peetje_{M_A}(a_i|\lambda)-\tfrac{1}{d}\right|\dee\mu_{\psi_d}(\lambda)
	&=
	\tfrac{1}{d}\int\left|d\peetje_{M_A}(a_i|\lambda)-\sum_{j=1}^d\peetje_{M_A}(a_j|\lambda)\right|\dee\mu_{\psi_d}(\lambda)\\
	&\leq
	\tfrac{1}{d}\sum_{j=1}^d\int\left|\peetje_{M_A}(a_i|\lambda)-\peetje_{M_A}(a_j|\lambda)\right|\dee\mu_{\psi_d}(\lambda).
\end{split}
\end{equation}
So it remains to be shown that
\begin{equation}\label{equibasic}
	\int\left|\peetje_{M_A}(a_i|\lambda)-\peetje_{M_A}(a_j|\lambda)\right|\dee\mu_{\psi_d}(\lambda)=0
\end{equation}
for all $i,j$.

Without loss of generality, it may be assumed that each $e_i$ is an eigenstate for $A$ for the eigenvalue $a_i$.
Let $B=\sum_{i=1}^db_i[e_i]$ be a second arbitrary complete self-adjoint operator and let $M_B$ be a local measurement corresponding to the operator $\een\otimes B$.

Now let $i,j$ be given.
If $i=j$, then \eqref{equibasic} holds trivially, so assume $i\neq j$.
For $\theta\in[0,\tfrac{\pi}{2}]$ the unitary operator $U_\theta$ is defined by the actions
\begin{equation}
\begin{gathered}
	U_\theta e_i=\cos\theta e_i+\sin\theta e_j,~ 
	U_\theta e_j=\cos\theta e_j-\sin\theta e_i,~
	U_\theta e_k=e_k~\text{when }k\neq i,j.
\end{gathered}
\end{equation}
Further define
\begin{equation}
	A_\theta\ldef U_\theta AU_\theta^*,~
	B_\theta\ldef U_\theta BU_\theta^*
\end{equation}
and let $M_{A_\theta}$ and $M_{B_\theta}$ denote local measurements corresponding to the operators $A_\theta\otimes\een$ and $\een\otimes B_\theta$ respectively.
For simplicity, it may be assumed that $M_{A_0}=M_A$, and the $M$ will be dropped in the notation.

For every $N\in\mathbb{N}$ the angle between $[0,\tfrac{\pi}{2}]$ can be divided in $2N+1$ equal sized smaller angles.
For $n\leq N$ set 
\begin{equation}
	A_{N,n}=A_{\frac{2n\pi}{(2N+1)2}},~
	B_{N,n}=B_{\frac{2(n+1)\pi}{(2N+1)2}}.
\end{equation}
This is used to make the estimate
\begin{multline}\label{YetAnotherEstimate}
	\int\left|\peetje_{A}(a_i|\lambda)-\peetje_{A}(a_j|\lambda)\right|\dee\mu_{\psi_d}(\lambda)\\
	\begin{aligned}
	\leq&
	\sum_{n=0}^{N}
	\int\left|\peetje_{A_{N,n}}(a_i|\lambda)-\peetje_{B_{N,n}}(b_i|\lambda)\right|\dee\mu_{\psi_d}(\lambda)
	+
	\sum_{n=0}^{N-1}
	\int\left|\peetje_{B_{N,n}}(b_i|\lambda)-\peetje_{A_{N,n+1}}(a_i|\lambda)\right|\dee\mu_{\psi_d}(\lambda)\\
	&+
	\int\left|\peetje_{B_{\frac{\pi}{2}}}(b_i|\lambda)-\peetje_{A}(a_j|\lambda)\right|\dee\mu_{\psi_d}(\lambda).
	\end{aligned}
\end{multline}

The final term in \eqref{YetAnotherEstimate} equals zero.
To show this, parameter independence is used to make the estimate
\begin{equation}
\begin{split}
	\peetje_{B_{\frac{\pi}{2}}}(b_i|\lambda)-\peetje_{A}(a_j|\lambda)
	&=
	\sum_{k=1}^d\peetje_{A\otimes B_{\frac{\pi}{2}}}(a_k,b_i|\lambda)-\sum_{l=1}^d\peetje_{A\otimes B_{\frac{\pi}{2}}}(a_j,b_l|\lambda)\\
	&=
	\sum_{k\neq j}\peetje_{A\otimes B_{\frac{\pi}{2}}}(a_k,b_i|\lambda)-\sum_{l\neq i}\peetje_{A\otimes B_{\frac{\pi}{2}}}(a_j,b_l|\lambda)\\
	&\leq
	\sum_{k\neq j}\peetje_{A\otimes B_{\frac{\pi}{2}}}(a_k,b_i|\lambda)+\sum_{l\neq i}\peetje_{A\otimes B_{\frac{\pi}{2}}}(a_j,b_l|\lambda).	
\end{split}
\end{equation}
Using this, it follows that
\begin{equation}
\begin{split}
	\int\left|\peetje_{B_{\frac{\pi}{2}}}(b_i|\lambda)-\peetje_{A}(a_j|\lambda)\right|\dee\mu_{\psi_d}(\lambda)
	\leq&
	\sum_{k\neq j}\int\peetje_{A\otimes B_{\frac{\pi}{2}}}(a_k,b_i|\lambda)\dee\mu_{\psi_d}(\lambda)
	+
	\sum_{l\neq i}\int\peetje_{A\otimes B_{\frac{\pi}{2}}}(a_j,b_l|\lambda)\dee\mu_{\psi_d}(\lambda)\\
	=&
	\sum_{k\neq j}\expv{\psi_d}{[e_k\otimes U_{\frac{\pi}{2}}e_i]}{\psi_d}
	+
	\sum_{l\neq i}\expv{\psi_d}{[U_{\frac{\pi}{2}}e_j\otimes e_l]}{\psi_d}\\
	=&
	\sum_{k\neq j}\expv{\psi_d}{[e_k\otimes e_j]}{\psi_d}
	+
	\sum_{l\neq i}\expv{\psi_d}{[e_i\otimes e_l]}{\psi_d}=0,
\end{split}
\end{equation}
where the final step follows from the perfect correlations in $\psi_d$.

For the other $2N+1$ terms in \eqref{YetAnotherEstimate} the correlations are not perfect.
But with parameter independence one still has for each term the estimate 
\begin{equation}
\begin{split}
	\int\left|\peetje_{A_\theta}(a_i|\lambda)-\peetje_{B_\phi}(b_i|\lambda)\right|\dee\mu_{\psi_d}(\lambda)
	\leq&
	\sum_{k\neq i}\int\peetje_{A_\theta\otimes B_\phi}(a_i,b_k|\lambda)\dee\mu_{\psi_d}(\lambda)
	+
	\sum_{l\neq i}\int\peetje_{A_\theta\otimes B_\phi}(a_l,b_i|\lambda)\dee\mu_{\psi_d}(\lambda)\\
	=&
	\sum_{k\neq i}\expv{\psi_d}{[U_\theta e_i\otimes U_\phi e_k]}{\psi_d}
	+
	\sum_{l\neq i}\expv{\psi_d}{[U_\theta e_l\otimes U_\phi e_i]}{\psi_d}\\
	=&
	\expv{\psi_d}{[U_\theta e_i\otimes U_\phi e_j]}{\psi_d}
	+
	\expv{\psi_d}{[U_\theta e_j\otimes U_\phi e_i]}{\psi_d}.
\end{split}
\end{equation}
And for the expectation values one has
\begin{equation}
\begin{split}
	\expv{\psi_d}{[U_\theta e_i\otimes U_\phi e_j]}{\psi_d}
	&=
	\frac{1}{d}\sum_{k=i,j}\sum_{l=i,j}
	\expv{e_k\otimes e_k}{[U_\theta e_i\otimes U_\phi e_j]}{e_l\otimes e_l}\\
	&=
	\frac{1}{d}\sum_{k=i,j}\sum_{l=i,j}
	\braket{e_k}{U_\theta e_i}\braket{e_k}{U_\phi e_j}
	\braket{U_\theta e_i}{e_l}\braket{U_\phi e_j}{e_l}\\
	&=
	\frac{1}{d}\sum_{k=i,j}
	\braket{e_k}{U_\theta e_i}\braket{e_k}{U_\phi e_j}\left(\cos\theta\sin\phi-\sin\theta\cos\phi\right)\\
	&=
	\frac{1}{d}\left(\cos\theta\sin\phi-\sin\theta\cos\phi\right)^2=\frac{1}{d}\sin^2(\phi-\theta).
\end{split}
\end{equation}

Making use of the fact that for all the $2N+1$ terms in \eqref{YetAnotherEstimate} $\left|\theta-\phi\right|=\tfrac{\pi}{(2N+1)2}$:
\begin{equation}
	\int\left|\peetje_{A}(a_i|\lambda)-\peetje_{A}(a_j|\lambda)\right|\dee\mu_{\psi_d}(\lambda)
	\leq
	(2N+1)\frac{2}{d}\sin^2\left(\frac{\pi}{(2N+1)2}\right)
	\leq
	\frac{\pi^2}{2d(2N+1)}.
\end{equation}
Since this holds for all $N\in\mathbb{N}$, the proof for \eqref{equibasic} is complete.
 
Finally, consider the case where $A$ is a non-maximal operator.
Assume again that each $e_i$ is an eigenstate for $A$.
Let $a_k$ be any eigenvalue for $A$ with degeneracy $n_k$ and complete set of eigenstates $e_{i_1},\ldots,e_{i_{n_k}}$. 
Let $B$ be a complete self-adjoint operator as before.
Then 
\begin{equation}
	\int\left|\peetje_{A}(a_k|\lambda)-\frac{n_k}{d}\right|\dee\mu_{\psi_d}(\lambda)
	\leq
	\int\left|\peetje_{A}(a_k|\lambda)-\sum_{j=1}^{n_k}\peetje_{B}(b_{i_j}|\lambda)\right|\dee\mu_{\psi_d}(\lambda)
	+
	\int\left|\sum_{j=1}^{n_k}\left(\peetje_{B}(b_{i_j}|\lambda)-\frac{1}{d}\right)\right|\dee\mu_{\psi_d}(\lambda).
\end{equation}
The first term equals zero because of the perfect correlations between the outcomes and that the second term equals zero follows by applying \eqref{equibasic} with $A$ replaced by $B$.
\end{proof}

\section{Proof of the extended equiprobability theorem}\label{appendix2}

\begin{theorem}
Consider the quantum system of a pair of $d$-level systems and the entangled state
\begin{equation}
	\psi_d^S\ldef\sum_{i=1}^{d}c_ie_i\otimes e_i\in\h_A\otimes\h_B\simeq\mathbb{C}^d\otimes\mathbb{C}^d.
\end{equation}
In a complete ontic model that satisfies ancilla independence and parameter independence, every ontic model for the system pair must be trivial with respect to all preparations that are represented by $\psi_d^S$ and all local measurements whose corresponding operator is diagonal in the orthonormal basis $(e_i)_{i=1}^d$.
\end{theorem}

\begin{proof}
\item
\paragraph{Setup}
The main strategy of the proof is to reshape the situation such that the proof strategy of the equiprobability theorem can be adopted.
There are then two entwined issues that need to be faced.

First, the proof of \autoref{equipthm} relies on the state being maximally entangled.
To solve this issue, one uses the method of embezzlement \citep{DamHayden03}.
The system is coupled to a second pair of systems prepared in a special entangled state.
Then, by applying only local unitary operations, one can transform the initial state to a state that is arbitrarily close to any desired entangled state.

The second problem is that the equiprobability theorem yields equal probabilities, whereas now arbitrary Born probabilities with values $\left|c_i\right|^2$ are to be derived.
This is solved by approximating these probabilities with sums of equal probabilities, which now pertain to yet another pair of systems coupled to the initial pair.

Without loss of generality, it may be assumed that the basis vectors $e_i$ are such that all the $c_i$ are positive reals.
Now let $M_A$ be a complete local measurement with corresponding operator $A\otimes \een$ and $A=\sum_{i=1}^da_i[e_i]$.
The aim is to show that for any $\Epsilon>0$
\begin{equation}\label{ToShow}
	\int\left|\peetje_{A}(a_i|\lambda)-c_i^2\right|\dee\mu_{\psi_d^S}(\lambda)<\Epsilon
\end{equation}
for all $i=1,\ldots,d$.

The first step is to approximate the Born probabilities with rational numbers.
For given $\epsilon>0$, choose natural numbers $n^\epsilon_1,\ldots,n^\epsilon_d$ such that\footnote{How $\epsilon$ is to relate to $\Epsilon$ is made clear at the end of the proof.}
\begin{equation}\label{RationalEst}
	\frac{c_i^2d^\epsilon}{n_i^\epsilon}\in(1-\epsilon,1+\epsilon),~\text{with }d^\epsilon\ldef\sum_{i=1}^dn_i^\epsilon,
\end{equation}
for all $i=1,\ldots,d$.
The system is then coupled to a pair of systems that is big enough to harbor this approximation.
The Hilbert spaces $\h_A^\epsilon,\h_B^\epsilon$ for these systems are copies of $\mathbb{C}^{m^\epsilon}$ with $m^\epsilon\ldef\max_{i}n_i^\epsilon$ and for each an orthonormal basis $f_1,f_2,\ldots,f_{m^\epsilon}$ is fixed.
So the idea is that because $\sum_{j=1}^{n^\epsilon_i}\tfrac{1}{d^\epsilon}\approx c_i^2$, the initial state $\psi_d^S\in\h_A\otimes\h_B$ can be mimicked by the maximally entangled state
\begin{equation}
	\psi_{d^\epsilon}\ldef\sum_{i=1}^d\sum_{j=1}^{n^\epsilon_i}\frac{1}{\sqrt{d^\epsilon}}f_{j}\otimes f_{j}\otimes e_i\otimes e_i,
\end{equation}
which is an element of $\h_A^\epsilon\otimes\h_B^\epsilon\otimes\h_A\otimes\h_B$.

For any results for the extended system to be relevant for the initial system, it has to evolve towards a state close to $\psi_{d^\epsilon}$.
This is accomplished by an interaction with a third pair of systems with corresponding Hilbert spaces $\h_A^N,\h_B^N$ that are copies of $\mathbb{C}^{D}$ where
\begin{equation}
	D\ldef N\prod_{i=1}^dn^\epsilon_i
\end{equation}
and $N$ is a natural number that in the end will be taken to be large enough to have the right relation to $\Epsilon$.

For both spaces an orthonormal basis $g_1,\ldots,g_{D}$ is fixed and it is assumed that the pair is initially prepared in the state
\begin{equation}
	\psi_{D}\ldef C_N\sum_{k=1}^{D}\frac{1}{\sqrt{k}}g_k\otimes g_k\in\h_A^N\otimes\h_B^N,
\end{equation}
where $C_N$ is a normalization constant.
The initial state for the totality of six systems is then taken to be
\begin{equation}\label{PsiI}
	\Psi_{\mathrm{I}}\ldef
	C_N\sum_{k=1}^{D}\sum_{i=1}^d\frac{c_i}{\sqrt{k}}
	g_k\otimes g_k\otimes f_{1}\otimes f_{1}\otimes e_i\otimes e_i.
\end{equation}
Note  that although $\Psi_{\mathrm{I}}$ depends on the parameters $\epsilon,N$, this is suppressed in the notation, as these parameters are taken to be fixed for the larger part of the proof.

Throughout, the bases will be fixed and a considerable chunk of the proof reduces to manipulations of the indices of these basis vectors.
Therefore it is useful to adopt a notation that highlights the indices:
\begin{equation}
	\bket{k}{\tilde{k}}{j}{\tilde{\jmath}}{i}{\tilde{\imath}}\ldef
	g_{k}\otimes g_{\tilde{k}}\otimes f_{j}\otimes f_{\tilde{\jmath}}\otimes e_{i}\otimes e_{\tilde{\imath}}.
\end{equation}
So \eqref{PsiI} becomes
\begin{equation}
	\Psi_{\mathrm{I}}=
	C_N\sum_{k=1}^{D}\sum_{i=1}^d\frac{c_i}{\sqrt{k}}
	\bket{k}{k}{1}{1}{i}{i}.
\end{equation}
In a similar fashion a more compact notation for products of projection operators is also adopted:
\begin{equation}
	\proj{P_1}{P_2}{P_3}{P_4}{P_5}{P_6}
	\ldef
	P_1\otimes P_2\otimes P_3\otimes P_4\otimes P_5\otimes P_6.
\end{equation}

\paragraph{Coupling the systems}
The first step is to use the fact that the ontic models are part of a complete ontic model.
So for every $N,\epsilon$ there exists a Markov kernel $\gamma_{N,\epsilon}$ that translates the probability measure $\mu_{\psi_d^S}$ for the initial ontic model to a probability measure 
\begin{equation}
	\mu_{\Psi_{\mathrm{I}}}(\Delta)\ldef\int\gamma_{N,\epsilon}(\Delta|\lambda)\dee\mu_{\psi_d^S}(\lambda)
\end{equation}
for the six systems, that corresponds to a preparation of the state $\Psi_{\mathrm{I}}$.
Moreover, because of ancilla independence, the initial response function $\peetje_A$ should correspond to some response function $\tilde{\peetje}_A$ in the larger ontic model such that
\begin{equation}
	\peetje_A(a_i|\lambda)=\int\tilde{\peetje}_A(a_i|\tilde{\lambda})\gamma_{N,\epsilon}(\dee\tilde{\lambda}|\lambda).
\end{equation}
One then has the estimate
\begin{equation}\label{CouplingStep}
\int\left|\peetje_{A}(a_i|\lambda)-c_i^2\right|\dee\mu_{\psi_d^S}(\lambda)
	=
	\int\left|\int\tilde{\peetje}_A(a_i|\tilde{\lambda})\gamma_{N,\epsilon}(\dee\tilde{\lambda}|\lambda)-c_i^2\right|\dee\mu_{\psi_d^S}(\lambda)
	\leq
	\int\left|\tilde{\peetje}_{A}(a_i|\tilde{\lambda})-c_i^2\right|\dee\mu_{\psi_{\mathrm{I}}}(\tilde{\lambda}).
\end{equation}
So the remainder of the proof focuses entirely on the large system.

\paragraph{Local measurements on the total system}
The main estimate is obtained by chaining Bell inequalities for local measurements on the total system.
The local measurements needed for the proof, however, can be conveniently defined in terms of the bases for $\h_A^\epsilon,\h_B^\epsilon,\h_A,\h_B$ alone.
Consider measurements $M_{\overline{A}},M_{\overline{B}}$ corresponding to the self-adjoint operators
\begin{equation}
\begin{aligned}
	\overline{A}
	\ldef&
	\sum_{i=1}^d\sum_{j=1}^{n_i^\epsilon}a_{i,j}
	U^*P^A_{i,j}U,~\text{where}~
	P^A_{i,j}
	\ldef\proj{\een}{\een}{[f_{j}]}{\een}{[e_i]}{\een},\\
	\overline{B}
	\ldef&
	\sum_{i=1}^d\sum_{j=1}^{n_i^\epsilon}b_{i,j}
	V^*P^B_{i,j}V,~\text{where}~
	P^B_{i,j}
	\ldef\proj{\een}{\een}{\een}{[f_{j}]}{\een}{[e_i]}.
\end{aligned}
\end{equation}
Here $U$ and $V$ are local unitary operators, i.e., $U$ acts as the unit on the Hilbert spaces $\h_B^N,\h_B^\epsilon,\h_B$ and  $V$ acts as the unit on the Hilbert spaces $\h_A^N,\h_A^\epsilon,\h_A$.
They are further assumed to have the following action on specific basis vectors:\footnote{Here $\ceil*{x}$ denotes the smallest natural number $n$ such that $n\geq x$.}
\begin{equation}
	U\bket{k}{k}{1}{1}{i}{i}
	=
	\bket{\ceil*{\tfrac{k}{n_i^\epsilon}}}{k}{j_{i,k}}{1}{i}{i},~
	V\bket{k}{k}{1}{1}{i}{i}
	=
	\bket{k}{\ceil*{\tfrac{k}{n_i^\epsilon}}}{1}{j_{i,k}}{i}{i},
\end{equation}
where
\begin{equation}
	j_{i,k}\ldef k-n_i^\epsilon\left(\ceil*{\tfrac{k}{n_i^\epsilon}}-1\right)
\end{equation}
takes on values in the range $\{1,2,\ldots,n_i^\epsilon\}$.
It is further convenient to introduce the state
\begin{equation}
	\Psi_{\mathrm{F}}
	\ldef UV\Psi_{\mathrm{I}}
	=
	C_N\sum_{k=1}^{D}\sum_{i=1}^d\frac{c_i}{\sqrt{k}}
	\bket{\ceil*{\tfrac{k}{n_i^\epsilon}}}{\ceil*{\tfrac{k}{n_i^\epsilon}}}{j_{i,k}}{j_{i,k}}{i}{i}.
\end{equation}
The theory of embezzlement states that this state becomes a better approximation of $\psi_D\otimes\psi_{d^\epsilon}$ as $N$ becomes larger and $\epsilon$ becomes smaller.
This is, on a formal level, what makes the proof work.

The next step is to make a further estimate picking up again from \eqref{CouplingStep}:
\begin{equation}\label{MainSplit}
\begin{split}
	\int\left|\tilde{\peetje}_{A}(a_i|\lambda)-c_i^2\right|\dee\mu_{\Psi_{\mathrm{I}}}(\lambda)
	\leq&
	\int\left|\tilde{\peetje}_{A}(a_i|\lambda)-\sum_{j=1}^{n_i^\epsilon}\peetje_{\overline{B}}(b_{i,j}|\lambda)\right|\dee\mu_{\Psi_{\mathrm{I}}}(\lambda)\\
	&+
	\int\left|\sum_{j=1}^{n_i^\epsilon}\peetje_{\overline{B}}(b_{i,j}|\lambda)-\frac{n_i^\epsilon}{d^\epsilon}\right|\dee\mu_{\Psi_{\mathrm{I}}}(\lambda)
	+
	\int\left|\frac{n_i^\epsilon}{d^\epsilon}-c_i^2\right|\dee\mu_{\Psi_{\mathrm{I}}}(\lambda).
\end{split}
\end{equation}
By \eqref{RationalEst} the final term is smaller than $\epsilon$.
The first term will be shown to equal zero.
Making use of parameter independence and the fact that $V$ is a local operator, the standard estimate gives
\begin{multline}
	\int\left|\tilde{\peetje}_{A}(a_i|\lambda)-\sum_{j=1}^{n^\epsilon_i}\peetje_{\overline{B}}(b_{i,j}|\lambda)\right|\dee\mu_{\Psi_{\mathrm{I}}}(\lambda)\\
	\begin{aligned}
		\leq&
		\sum_{\tilde{\imath}\neq i}\sum_{j=1}^{n_{\tilde{\imath}}^\epsilon}
		\int\peetje_{A,\overline{B}}(a_i,b_{\tilde{\imath},j}|\lambda)\dee\mu_{\Psi_{\mathrm{I}}}(\lambda)
		+
		\sum_{\tilde{\imath}'\neq i}\sum_{j=1}^{n_i^\epsilon}
		\int\peetje_{A,\overline{B}}(a_{\tilde{\imath}'},b_{i,j}|\lambda)\dee\mu_{\Psi_{\mathrm{I}}}(\lambda)\\
		=&
		\sum_{\tilde{\imath}\neq i}\sum_{j=1}^{n_{\tilde{\imath}}^\epsilon}
		\expv{V\Psi_{\mathrm{I}}}%
		{\proj{\een}{\een}{\een}{[f_{j}]}{[e_i]}{[e_{\tilde{\imath}}]}}%
		{V\Psi_{\mathrm{I}}}
		+
		\sum_{\tilde{\imath}'\neq i}\sum_{j=1}^{n_i^\epsilon}
		\expv{V\Psi_{\mathrm{I}}}%
		{\proj{\een}{\een}{\een}{[f_{j}]}{[e_{\tilde{\imath}'}]}{[e_i]}}%
		{V\Psi_{\mathrm{I}}}.
	\end{aligned}
\end{multline}
Then each term in the final sum equals zero:
\begin{multline}
	\expv{V\Psi_{\mathrm{I}}}%
	{\proj{\een}{\een}{\een}{[f_l]}{[e_m]}{[e_n]}}%
	{V\Psi_{\mathrm{I}}}\\
	\begin{aligned}
		=&
		C_N^2
		\sum_{k,\tilde{k}=1}^{D}\sum_{i,\tilde{\imath}=1}^d
		\sum_{j=1}^{n_i^\epsilon}\sum_{\tilde{\jmath}=1}^{n_{\tilde{\imath}}^\epsilon}
		\frac{c_ic_{\tilde{\imath}}}{\sqrt{k\tilde{k}}}
		\braket{g_k}{g_{\tilde{k}}}\braket{g_{\ceil{\frac{k}{n_i^\epsilon}}}}{g_{\ceil{\frac{\tilde{k}}{n_{\tilde{\imath}}^\epsilon}}}}
		\braket{f_1}{f_1}
		\expv{f_{j_{i,k}}}{[f_l]}{f_{j_{\tilde{\imath},\tilde{k}}}}
		\expv{e_i}{[e_m]}{e_{\tilde{\imath}}}\expv{e_i}{[e_n]}{e_{\tilde{\imath}}}\\
		=&
		0,
	\end{aligned}
\end{multline}
where the final equality follows because $m\neq n$.

The conclusion of this step is then that
\begin{equation}\label{IntermedResult}
	\int\left|\tilde{\peetje}_{A}(a_i|\lambda)-c_i^2\right|\dee\mu_{\Psi_{\mathrm{I}}}(\lambda)
	<
	\int\left|\sum_{j=1}^{n_i^\epsilon}\peetje_{\overline{B}}(b_{i,j}|\lambda)-\frac{n_i^\epsilon}{d^\epsilon}\right|\dee\mu_{\Psi_{\mathrm{I}}}(\lambda)
	+\epsilon.
\end{equation}

\paragraph{Applying chained Bell inequalities}
The next step is to find an estimate for the second term on the right hand side of \eqref{MainSplit}.
This is where the strategy for the proof of \autoref{equipthm} is copied.
First note that
\begin{equation}\label{Chainz}
\begin{split}
	\int\left|\sum_{j=1}^{n_i^\epsilon}\peetje_{\overline{B}}(b_{i,j}|\lambda)-\frac{n_i^\epsilon}{d^\epsilon}\right|\dee\mu_{\Psi_{\mathrm{I}}}(\lambda)
	\leq&
	\sum_{j=1}^{n_i^\epsilon}\int\left|\peetje_{\overline{B}}(b_{i,j}|\lambda)-\frac{1}{d^\epsilon}\right|\dee\mu_{\Psi_{\mathrm{I}}}(\lambda)\\
	\leq&
	\frac{1}{d^\epsilon}
	\sum_{j=1}^{n_i^\epsilon}\sum_{\tilde{\imath}=1}^{d}\sum_{\tilde{\jmath}=1}^{n_{\tilde{\imath}}^\epsilon}
	\int\left|\peetje_{\overline{B}}(b_{i,j}|\lambda)-\peetje_{\overline{B}}(b_{\tilde{\imath},\tilde{\jmath}}|\lambda)\right|\dee\mu_{\Psi_{\mathrm{I}}}(\lambda).
\end{split}
\end{equation}

Next, for fixed values of $i,j,\tilde{\imath},\tilde{\jmath}$ and for $\theta\in[0,\pi/2]$, two local unitary operators $U_\theta$ and $V_\theta$ are defined with the following actions:\footnote{For notational convenience the variables $i,j,\tilde{\imath},\tilde{\jmath}$ are suppressed in the expression for these unitary operators.}
\begin{equation}
\begin{aligned}
	U_\theta\bket{k_1}{k_2}{j}{j_2}{i}{i_2}
	&=
	\cos\theta\bket{k_1}{k_2}{j}{j_2}{i}{i_2}
	+\sin\theta\bket{k_1}{k_2}{\tilde{\jmath}}{j_2}{\tilde{\imath}}{i_2},\\
	U_\theta\bket{k_1}{k_2}{\tilde{\jmath}}{j_2}{\tilde{\imath}}{i_2}
	&=
	\cos\theta\bket{k_1}{k_2}{\tilde{\jmath}}{j_2}{\tilde{\imath}}{i_2}
	-\sin\theta\bket{k_1}{k_2}{j}{j_2}{i}{i_2}
\end{aligned}
\end{equation}
and
\begin{equation}
\begin{aligned}
	V_\theta\bket{k_1}{k_2}{j_1}{j}{i_1}{i}
	&=
	\cos\theta\bket{k_1}{k_2}{j_1}{j}{i_1}{i}
	+\sin\theta\bket{k_1}{k_2}{j_1}{\tilde{\jmath}}{i_1}{\tilde{\imath}},\\
	V_\theta\bket{k_1}{k_2}{j_1}{\tilde{\jmath}}{i_1}{\tilde{\imath}}
	&=
	\cos\theta\bket{k_1}{k_2}{j_1}{\tilde{\jmath}}{i_1}{\tilde{\imath}}
	-\sin\theta\bket{k_1}{k_2}{j_1}{j}{i_1}{i},
\end{aligned}
\end{equation}
for arbitrary values of $k_1,k_2,j_1,j_2,i_1,i_2$ and $U_\theta$ and $V_\theta$ act as the unit operator on all other basis vectors.

For given $i,j,\tilde{\imath},\tilde{\jmath},\theta$, two local measurements $M_{\overline{A}_\theta}$ and $M_{\overline{B}_\theta}$ are introduced corresponding to the operators 
\begin{equation}
	\overline{A}_\theta
	\ldef
	\sum_{\iota=1}^d\sum_{\eta=1}^{n_\iota^\epsilon}a_{\iota,\eta}
	U^*U_\theta^*P^A_{\iota,\eta}U_\theta U,~
	\overline{B}_\theta
	\ldef
	\sum_{\iota=1}^d\sum_{\eta=1}^{n_\iota^\epsilon}b_{\iota,\eta}
	V^*V_\theta^*P^B_{\iota,\eta}V_\theta V.
\end{equation}
Without loss of generality, it may be assumed that $M_{\overline{B}_0}$ and $M_{\overline{B}}$ refer to the same measurement.

For any $L\in\mathbb{N}$ the angle from 0 to $\tfrac{\pi}{2}$ can be divided into $2L+1$ equal smaller angles so as to obtain the inequality
\begin{equation}\label{MoreChains}
\begin{split}
	\int\left|\peetje_{\overline{B}}(b_{i,j}|\lambda)-\peetje_{\overline{B}}(b_{\tilde{\imath},\tilde{\jmath}}|\lambda)\right|\dee\mu_{\Psi_\mathrm{I}}(\lambda)
	\leq&
	\sum_{l=0}^L
	\int\left|\peetje_{\overline{B}_{L,l}}(b_{i,j}|\lambda)
	-\peetje_{\overline{A}_{L,l}}(a_{i,j}|\lambda)\right|\dee\mu_{\Psi_\mathrm{I}}(\lambda)\\
	&+
	\sum_{l=0}^{L-1}
	\int\left|\peetje_{\overline{A}_{L,l}}(a_{i,j}|\lambda)
	-\peetje_{\overline{B}_{L,l+1}}(b_{i,j}|\lambda)\right|\dee\mu_{\Psi_\mathrm{I}}(\lambda)\\
	&+
	\int\left|\peetje_{\overline{A}_{L,L}}(a_{i,j}|\lambda)
	-\peetje_{\overline{B}}(b_{\tilde{\imath},\tilde{\jmath}}|\lambda)\right|\dee\mu_{\Psi_\mathrm{I}}(\lambda),
\end{split}
\end{equation}
where
\begin{equation}
	\overline{B}_{L,l}\ldef\overline{B}_{\frac{2l\pi}{(2L+1)2}}
	\text{ and }
	\overline{A}_{L,l}\ldef\overline{A}_{\frac{(2l+1)\pi}{(2L+1)2}}.
\end{equation}
Evaluating each of these terms is somewhat tedious.
Therefore the calculations to obtain the following results are transported to the next section.
For the final term one simply has
\begin{equation}\label{ZeroEq}
	\int\left|\peetje_{\overline{A}_{\frac{\pi}{2}}}(a_{i,j}|\lambda)
		-\peetje_{\overline{B}}(b_{\tilde{\imath},\tilde{\jmath}}|\lambda)\right|\dee\mu_{\Psi_\mathrm{I}}(\lambda)=0.
\end{equation}
For the other terms one obtains the expression
\begin{equation}\label{TediousExpr}
\begin{split}
	\int\left|\peetje_{\overline{A}_\theta}(a_{i,j}|\lambda)
	-\peetje_{\overline{B}_\phi}(b_{i,j}|\lambda)\right|\dee\mu_{\Psi_\mathrm{I}}(\lambda)
	\leq&
	\frac{C_N^2}{2}\sum_{m=1}^{\frac{D}{n^\epsilon_i}}
	\left(\frac{c_i^2(2\sin^2(\theta-\phi)+\sin2\theta\sin2\phi)}{k_{m,i,j}} 
	-
	\frac{c_ic_{\tilde{\imath}}\sin2\theta\sin2\phi}{\sqrt{k_{m,i,j}k_{m,\tilde{\imath},\tilde{\jmath}}}}
	\right)\\
	&+
	\frac{C_N^2}{2}\sum_{m=1}^{\frac{D}{n^\epsilon_{\tilde{\imath}}}}
	\left(
	\frac{c_{\tilde{\imath}}^2(2\sin^2(\theta-\phi)+\sin2\theta\sin2\phi)}{k_{m,\tilde{\imath},\tilde{\jmath}}}
	-
	\frac{c_ic_{\tilde{\imath}}\sin2\theta\sin2\phi}{\sqrt{k_{m,i,j}k_{m,\tilde{\imath},\tilde{\jmath}}}}
	\right),
\end{split}
\end{equation}
where
\begin{equation}
	k_{m,i,j}\ldef j+(m-1)n_i^\epsilon.
\end{equation}

\paragraph{Obtaining the final estimate}
The sum over the angles in \eqref{MoreChains} is now postponed to first perform the sums over $j$ and $\tilde{\jmath}$ from \eqref{Chainz}.
For two of the four terms in \eqref{TediousExpr} this summation easily yields a nice result:
\begin{equation}
	\sum_{j=1}^{n^\epsilon_i}\sum_{\tilde{\jmath}=1}^{n^\epsilon_{\tilde{\imath}}}\sum_{m=1}^{\frac{D}{n_i^\epsilon}}\frac{1}{k_{m,i,j}}
	=
	\sum_{\tilde{\jmath}=1}^{n^\epsilon_{\tilde{\imath}}}\sum_{k=1}^{D}\frac{1}{k}
	=\frac{n^\epsilon_{\tilde{\imath}}}{C_N^2},~
	\sum_{j=1}^{n^\epsilon_i}\sum_{\tilde{\jmath}=1}^{n^\epsilon_{\tilde{\imath}}}\sum_{m=1}^{\frac{D}{n_{\tilde{\imath}}^\epsilon}}\frac{1}{k_{m,\tilde{\imath},\tilde{\jmath}}}
	=
	\sum_{j=1}^{n^\epsilon_{i}}\sum_{k=1}^{D}\frac{1}{k}
	=\frac{n^\epsilon_{i}}{C_N^2}.
\end{equation}
For the other two terms one has the estimate
\begin{equation}
\begin{split}
	\sum_{j=1}^{n^\epsilon_i}\sum_{\tilde{\jmath}=1}^{n^\epsilon_{\tilde{\imath}}}\sum_{m=1}^{\frac{D}{n_i^\epsilon}}
	\frac{1}{\sqrt{k_{m,i,j}k_{m,\tilde{\imath},\tilde{\jmath}}}}
	\geq&
	\sum_{j=1}^{n^\epsilon_i}\sum_{\tilde{\jmath}=1}^{n^\epsilon_{\tilde{\imath}}}\sum_{m=1}^{\frac{D}{n_i^\epsilon}}
	\frac{1}{\sqrt{mn^\epsilon_imn^\epsilon_{\tilde{\imath}}}}
	=
	\sqrt{n^\epsilon_in^\epsilon_{\tilde{\imath}}}\sum_{m=1}^{\frac{D}{n_i^\epsilon}}\frac{1}{m}
	=
	\sqrt{n^\epsilon_in^\epsilon_{\tilde{\imath}}}\left(\frac{1}{C_N^2}-\sum_{m=\frac{D}{n_i^\epsilon}+1}^D\frac{1}{m}\right)\\
	\geq&
	\sqrt{n^\epsilon_in^\epsilon_{\tilde{\imath}}}\left(\frac{1}{C_N^2}-\int_{\frac{D}{n_i^\epsilon}}^D\frac{1}{x}\dee x\right)
	=
	\sqrt{n^\epsilon_in^\epsilon_{\tilde{\imath}}}\left(\frac{1}{C_N^2}-\log\left(n^\epsilon_i\right)\right)
\end{split}
\end{equation}
and similarly
\begin{equation}
	\sum_{j=1}^{n^\epsilon_i}\sum_{\tilde{\jmath}=1}^{n^\epsilon_{\tilde{\imath}}}\sum_{m=1}^{\frac{D}{n_{\tilde{\imath}}^\epsilon}}
	\frac{1}{\sqrt{k_{m,i,j}k_{m,\tilde{\imath},\tilde{\jmath}}}}
	\geq
	\sqrt{n^\epsilon_in^\epsilon_{\tilde{\imath}}}\left(\frac{1}{C_N^2}-\log\left(n^\epsilon_{\tilde{\imath}}\right)\right).
\end{equation}

Combining these estimates with the expression \eqref{TediousExpr} and making use of \eqref{RationalEst} gives
\begin{multline}
	\sum_{j=1}^{n^\epsilon_i}\sum_{\tilde{\jmath}=1}^{n^\epsilon_{\tilde{\imath}}}
	\int\left|\peetje_{\overline{A}_\theta}(a_{i,j}|\lambda)
	-\peetje_{\overline{B}_\phi}(b_{i,j}|\lambda)\right|\dee\mu_{\Psi_I^\epsilon}(\lambda)\\
	\begin{aligned}
	\leq&
	\frac{c_i^2n^\epsilon_{\tilde{\imath}}}{2}(2\sin^2(\theta-\phi)+\sin2\theta\sin2\phi)
	-
	\frac{c_ic_{\tilde{\imath}}\sqrt{n^\epsilon_in^\epsilon_{\tilde{\imath}}}}{2}\left(1-C_N^2\log\left(n^\epsilon_i\right)\right)
	\sin2\theta\sin2\phi\\
	&+
	\frac{c_{\tilde{\imath}}^2n^\epsilon_i}{2}(2\sin^2(\theta-\phi)+\sin2\theta\sin2\phi)
	-
	\frac{c_ic_{\tilde{\imath}}\sqrt{n^\epsilon_in^\epsilon_{\tilde{\imath}}}}{2}\left(1-C_N^2\log\left(n^\epsilon_{\tilde{\imath}}\right)\right)
	\sin2\theta\sin2\phi\\
	<&
	\frac{n^\epsilon_in^\epsilon_{\tilde{\imath}}(1+\epsilon)}{d^\epsilon}\left(2\sin^2(\theta-\phi)+\sin2\theta\sin2\phi\right)
	-
	\frac{n^\epsilon_in^\epsilon_{\tilde{\imath}}(1-\epsilon)}{d^\epsilon}\left(1-\tfrac{1}{2}C_N^2\log(n^\epsilon_in^\epsilon_{\tilde{\imath}})\right)
	\sin2\theta\sin2\phi\\
	<&
	2n^\epsilon_{\tilde{\imath}}\sin^2(\theta-\phi)+2n^\epsilon_{\tilde{\imath}}\epsilon+n^\epsilon_{\tilde{\imath}}C_N^2\log(d^\epsilon).
	\end{aligned}
\end{multline} 

Finally, combining this result with \eqref{IntermedResult}, \eqref{Chainz} and \eqref{MoreChains}, one obtains
\begin{equation}\label{FinalFinal}
\begin{split}
	\int\left|
	\tilde{\peetje}_A(a_i|\lambda)-c_i^2
	\right|\dee\mu_{\psi_\mathrm{I}}(\lambda)
	<&
	\epsilon+
	\frac{1}{d^\epsilon}\sum_{\tilde{\imath} =1}^d(2L+1)n^\epsilon_{\tilde{\imath}}\left(
	2\sin^2\left(\tfrac{\pi}{2(2L+1)}\right)+2\epsilon+C_N^2\log(d^\epsilon)
	\right)\\
	<&
	\frac{\pi^2}{2(2L+1)}
	+
	(4L+3)\epsilon
	+
	\frac{(2L+1)\log(d^\epsilon)}{\log(1+N\prod_{i=1}^dn^\epsilon_i)}.
\end{split}
\end{equation}
Now choose $L$ such that $\tfrac{\pi^2}{2(2L+1)}<\tfrac{1}{3}\Epsilon$, then choose $\epsilon$ such that $(4L+3)\epsilon<\tfrac{1}{3}\Epsilon$, and finally choose $N$ such that $\tfrac{(2L+1)\log(d^\epsilon)}{\log(1+N\prod_{i=1}^dn^\epsilon_i)}<\tfrac{1}{3}\Epsilon$.
Plugging these choices into \eqref{FinalFinal} one obtains \eqref{ToShow}.

\end{proof}

\section{Some remaining calculations}
\paragraph{Proof for \eqref{ZeroEq}}
The equality follows with the use of parameter independence to get an estimate in terms of quantum probabilities.
\begin{multline}
	\int\left|\peetje_{\overline{A}_{\frac{\pi}{2}}}(a_{i,j}|\lambda)
		-\peetje_{\overline{B}}(b_{\tilde{\imath},\tilde{\jmath}}|\lambda)\right|\dee\mu_{\Psi_\mathrm{I}}(\lambda)\\
		\begin{aligned}
		\leq&
		\sum_{(r,s)\neq(\tilde{\imath},\tilde{\jmath})}\int\peetje_{\overline{A}_{\frac{\pi}{2}}\overline{B}}(a_{i,j}b_{r,s}|\lambda)\dee\mu_{\Psi_\mathrm{I}}(\lambda)
		+
		\sum_{(r,s)\neq(i,j)}\int\peetje_{\overline{A}_{\frac{\pi}{2}}\overline{B}}(a_{r,s}b_{\tilde{\imath},\tilde{\jmath}}|\lambda)\dee\mu_{\Psi_\mathrm{I}}(\lambda)\\
		=&
		\sum_{(r,s)\neq(\tilde{\imath},\tilde{\jmath})}\expv{\Psi_\mathrm{F}}{U_{\frac{\pi}{2}}^*P^A_{i,j}U_{\frac{\pi}{2}}P^B_{r,s}}{\Psi_\mathrm{F}}
		+
		\sum_{(r,s)\neq(i,j)}\expv{\Psi_\mathrm{F}}{U_{\frac{\pi}{2}}^*P^A_{r,s}U_{\frac{\pi}{2}}P^B_{\tilde{\imath},\tilde{\jmath}}}{\Psi_\mathrm{F}}\\
		=&
		\sum_{(r,s)\neq(\tilde{\imath},\tilde{\jmath})}\expv{\Psi_\mathrm{F}}{P^A_{\tilde{\imath},\tilde{\jmath}}P^B_{r,s}}{\Psi_\mathrm{F}}
		+
		\sum_{(r,s)\neq(\tilde{\imath},\tilde{\jmath})}\expv{\Psi_\mathrm{F}}{P^A_{r,s}P^B_{\tilde{\imath},\tilde{\jmath}}}{\Psi_\mathrm{F}}
		=0,
		\end{aligned}
\end{multline}
where the final step follows because of the perfect correlations in $\Psi_\mathrm{F}$.

\paragraph{Proof for \eqref{TediousExpr}}
The first step is again to make use of parameter independence to get an estimate in terms of quantum probabilities.
\begin{multline}\label{FirstStepTedious}
	\int\left|\peetje_{\overline{A}_\theta}(a_{i,j}|\lambda)
		-\peetje_{\overline{B}_\phi}(b_{i,j}|\lambda)\right|\dee\mu_{\Psi_\mathrm{I}}(\lambda)\\
		\begin{aligned}
		\leq&
		\sum_{(r,s)\neq(i,j)}\int\peetje_{\overline{A}_\theta\overline{B}_\phi}(a_{i,j}b_{r,s}|\lambda)\dee\mu_{\Psi_\mathrm{I}}(\lambda)
		+
		\sum_{(r,s)\neq(i,j)}\int\peetje_{\overline{A}_\theta\overline{B}_\phi}(a_{r,s}b_{i,j}|\lambda)\dee\mu_{\Psi_\mathrm{I}}(\lambda)\\
		=&
		\sum_{(r,s)\neq(i,j)}\expv{U_\theta V_\phi\Psi_\mathrm{F}}{P^A_{i,j}P^B_{r,s}}{U_\theta V_\phi\Psi_\mathrm{F}}
		+
		\sum_{(r,s)\neq(i,j)}\expv{U_\theta V_\phi\Psi_\mathrm{F}}{P^A_{r,s}P^B_{i,j}}{U_\theta V_\phi\Psi_\mathrm{F}}.
		\end{aligned}
\end{multline}
The next step is to find a better expression for the first sum, and then use the similarity between the two sums to also obtain an expression for the second sum. 
For this look at the action of the projection operators:
\begin{equation}\label{StateEval}
	P^A_{i,j}U_\theta P^B_{r,s}V_\phi\Psi_\mathrm{F}
	=
	C_N\sum_{k=1}^D\sum_{\iota=1}^d\frac{c_\iota}{\sqrt{k}}
	P^A_{i,j}U_\theta P^B_{r,s}V_\phi
	\bket{\ceil*{\tfrac{k}{n_\iota^\epsilon}}}{\ceil*{\tfrac{k}{n_\iota^\epsilon}}}{j_{\iota,k}}{j_{\iota,k}}{\iota}{\iota}.
\end{equation}
The action of $P^A_{i,j}U_\theta$ on each component in the sum yields a non-zero result only if $(\iota,j_{\iota,k})=(i,j)$ or $(\iota,j_{\iota,k})=(\tilde{\imath},\tilde{\jmath})$.
For $\iota=i$, $j_{\iota,k}$ takes on the value $j$ exactly $\tfrac{D}{n^\epsilon_i}$ times in the sum over $k$.
On these occasions, $k$ takes on the values $j,j+n^\epsilon_i,\ldots,j+(\tfrac{D}{n^\epsilon_i}-1)n^\epsilon_i$.
For this reason the expression $k_{m,i,j}=j+(m-1)n^\epsilon_i$ is introduced.
The case for $\iota=\tilde{\imath}$ is similar.
Thus performing the sum over $\iota$ in \eqref{StateEval} gives
\begin{equation}
	P^A_{i,j}U_\theta P^B_{r,s}V_\phi\Psi_\mathrm{F}
	=
	C_N\sum_{m=1}^{\frac{D}{n^\epsilon_i}}
	\frac{c_i\cos\theta}{\sqrt{k_{m,i,j}}}P^B_{r,s}V_\phi
	\bket{m}{m}{j}{j}{i}{i}
	-
	C_N\sum_{m=1}^{\frac{D}{n^\epsilon_{\tilde{\imath}}}}
	\frac{c_{\tilde{\imath}}\sin\theta}{\sqrt{k_{m,\tilde{\imath},\tilde{\jmath}}}}P^B_{r,s}V_\phi
	\bket{m}{m}{j}{\tilde{\jmath}}{i}{\tilde{\imath}}
\end{equation}  

Because of the action of the projection operator $P^B_{r,s}$, the sum over $(r,s)\neq(i,j)$ only picks up a term when $(r,s)=(\tilde{\imath},\tilde{\jmath})$.
So
\begin{equation}
\begin{split}
	\sum_{(r,s)\neq(i,j)}P^A_{i,j}U_\theta P^B_{r,s}V_\phi\Psi_\mathrm{F}
	=&
	C_N\sum_{m=1}^{\frac{D}{n^\epsilon_i}}
	\frac{c_i\cos\theta}{\sqrt{k_{m,i,j}}}P^B_{\tilde{\imath},\tilde{\jmath}}V_\phi
	\bket{m}{m}{j}{j}{i}{i}
	-
	C_N\sum_{m=1}^{\frac{D}{n^\epsilon_{\tilde{\imath}}}}
	\frac{c_{\tilde{\imath}}\sin\theta}{\sqrt{k_{m,\tilde{\imath},\tilde{\jmath}}}}P^B_{\tilde{\imath},\tilde{\jmath}}V_\phi
	\bket{m}{m}{j}{\tilde{\jmath}}{i}{\tilde{\imath}}\\
	=&
	C_N\sum_{m=1}^{\frac{D}{n^\epsilon_i}}
	\frac{c_i\cos\theta\sin\phi}{\sqrt{k_{m,i,j}}}
	\bket{m}{m}{j}{\tilde{\jmath}}{i}{\tilde{\imath}}
	-
	C_N\sum_{m=1}^{\frac{D}{n^\epsilon_{\tilde{\imath}}}}
	\frac{c_{\tilde{\imath}}\sin\theta\cos\phi}{\sqrt{k_{m,\tilde{\imath},\tilde{\jmath}}}}
	\bket{m}{m}{j}{\tilde{\jmath}}{i}{\tilde{\imath}}.
\end{split}
\end{equation}
Bringing also the other rotations in \eqref{FirstStepTedious} to the right gives
\begin{equation}
\begin{split}
	\sum_{(r,s)\neq(i,j)}U^*_\theta P^A_{i,j}U_\theta V^*_\phi P^B_{r,s}V_\phi\Psi_\mathrm{F}
	=&
	C_N\sum_{m=1}^{\frac{D}{n^\epsilon_i}}
	\frac{c_i\cos\theta\sin\phi}{\sqrt{k_{m,i,j}}}\left(\cos\theta V_\phi^*\bket{m}{m}{j}{\tilde{\jmath}}{i}{\tilde{\imath}}
	-\sin\theta V_\phi^*\bket{m}{m}{\tilde{\jmath}}{\tilde{\jmath}}{\tilde{\imath}}{\tilde{\imath}}
	\right)\\
	&+
	C_N\sum_{m=1}^{\frac{D}{n^\epsilon_{\tilde{\imath}}}}
	\frac{c_{\tilde{\imath}}\sin\theta\cos\phi}{\sqrt{k_{m,\tilde{\imath},\tilde{\jmath}}}}\left(\sin\theta V_\phi^*\bket{m}{m}{\tilde{\jmath}}{\tilde{\jmath}}{\tilde{\imath}}{\tilde{\imath}}
	-\cos\theta V_\phi^*\bket{m}{m}{j}{\tilde{\jmath}}{i}{\tilde{\imath}}
	\vphantom{\sin^2\theta}\right).
\end{split}
\end{equation}
And finally
\begin{multline}\label{TediousExpr2}
	\sum_{(r,s)\neq(i,j)}U^*_\theta P^A_{i,j}U_\theta V^*_\phi P^B_{r,s}V_\phi\Psi_\mathrm{F}\\
	\begin{aligned}
	=&
	C_N\sum_{m=1}^{\frac{D}{n^\epsilon_i}}
	\frac{c_i\cos\theta\sin\phi}{\sqrt{k_{m,i,j}}}\left(
	\cos\theta\cos\phi\bket{m}{m}{j}{\tilde{\jmath}}{i}{\tilde{\imath}}
	+\cos\theta\sin\phi\bket{m}{m}{j}{j}{i}{i}\right.\\
	&\left.
	\hphantom{C_N\sum_{m=1}^{\frac{D}{n^\epsilon_i}}\frac{c_i\cos\theta\sin\phi}{\sqrt{k_{m,i,j}}}(}
	-\sin\theta\cos\phi\bket{m}{m}{\tilde{\jmath}}{\tilde{\jmath}}{\tilde{\imath}}{\tilde{\imath}}
	-\sin\theta\sin\phi\bket{m}{m}{\tilde{\jmath}}{j}{\tilde{\imath}}{i}\right)\\
	&+
	C_N\sum_{m=1}^{\frac{D}{n^\epsilon_{\tilde{\imath}}}}
	\frac{c_{\tilde{\imath}}\sin\theta\cos\phi}{\sqrt{k_{m,\tilde{\imath},\tilde{\jmath}}}}\left(
	\sin\theta\cos\phi\bket{m}{m}{\tilde{\jmath}}{\tilde{\jmath}}{\tilde{\imath}}{\tilde{\imath}}
	+\sin\theta\sin\phi\bket{m}{m}{\tilde{\jmath}}{j}{\tilde{\imath}}{i}\right.\\
	&\left.
	\hphantom{+C_N\sum_{m=1}^{\frac{D}{n^\epsilon_{\tilde{\imath}}}}\frac{c_{\tilde{\imath}}\sin\theta\cos\phi}{\sqrt{k_{m,\tilde{\imath},\tilde{\jmath}}}}(}
	-\cos\theta\cos\phi\bket{m}{m}{j}{\tilde{\jmath}}{i}{\tilde{\imath}}
	-\cos\theta\sin\phi\bket{m}{m}{j}{j}{i}{i}\right)
	\end{aligned}
\end{multline}

Because of the perfect correlations in $\Psi_\mathrm{F}$, only the four terms with $i,i$ and $\tilde{\imath},\tilde{\imath}$ in the final expression contribute to the sum over expectation values in \eqref{FirstStepTedious}.
This results in
\begin{multline}
	\sum_{(r,s)\neq(i,j)}
	\expv{U_\theta V_\phi\Psi_\mathrm{F}}{P^A_{i,j}P^B_{r,s}}{U_\theta V_\phi\Psi_\mathrm{F}}\\
	\begin{aligned}
	=&
	C_N^2\sum_{m=1}^{\frac{D}{n^\epsilon_i}}
	\frac{c_i^2}{k_{m,i,j}}\cos^2\theta\sin^2\phi
	-
	C_N^2\sum_{m=1}^{\frac{D}{n^\epsilon_i}}
	\frac{c_ic_{\tilde{\imath}}}{\sqrt{k_{m,i,j}k_{m,\tilde{\imath},\tilde{\jmath}}}}\cos\theta\sin\theta\cos\phi\sin\phi\\
	&+
	C_N^2\sum_{m=1}^{\frac{D}{n^\epsilon_{\tilde{\imath}}}}
	\frac{c_{\tilde{\imath}}^2}{k_{m,\tilde{\imath},\tilde{\jmath}}}\sin^2\theta\cos^2\phi
	-
	C_N^2\sum_{m=1}^{\frac{D}{n^\epsilon_{\tilde{\imath}}}}
	\frac{c_ic_{\tilde{\imath}}}{\sqrt{k_{m,i,j}k_{m,\tilde{\imath},\tilde{\jmath}}}}\cos\theta\sin\theta\cos\phi\sin\phi.
	\end{aligned}
\end{multline}
Making use of the trigonometric expression $\sin^2(\theta-\phi)=\cos^2\theta\sin^2\phi+\sin^2\theta\cos^2\phi-2\cos\theta\sin\theta\cos\phi\sin\phi$, the desired result \eqref{TediousExpr} follows.
\end{appendix}

\printbibliography

\end{document}